%% file: main.tex
\documentclass[11pt,notitlepage,tightenlines,nofootinbib,superscriptaddress]{revtex4-1}

\usepackage[utf8]{inputenc}
\input{macros}

\usepackage{tabularx}
\usepackage{makecell}

\begin{document}

\title{Tumula information\\ and doubly minimized Petz R\'enyi lautum information}

\author{Lukas Schmitt}
\affiliation{Institute for Theoretical Physics, ETH Zurich}\affiliation{IBM Quantum, IBM Research Europe – Zurich}

\author{Filippo Girardi}
\affiliation{Scuola Normale Superiore (SNS)}

\author{Laura Burri}
\affiliation{Institute for Theoretical Physics, ETH Zurich}

\begin{abstract}
We study a doubly minimized variant of the lautum information -- a reversed analogue of the mutual information -- defined as the minimum relative entropy between any product state and a fixed bipartite quantum state; we refer to this measure as the tumula information. In addition, we introduce the corresponding Petz R\'enyi version, which we call the doubly minimized Petz R\'enyi lautum information (PRLI). 
We derive several general properties of these correlation measures and provide an operational interpretation in the context of hypothesis testing. 
Specifically, we show that the reverse direct exponent of certain binary quantum state discrimination problems is quantified by the doubly minimized PRLI of order $\alpha\in (0,1/2)$, and that the Sanov exponent is determined by the tumula information. 
Furthermore, we investigate the extension of the tumula information to channels and compare its properties with previous results on the channel umlaut information.
\end{abstract}

\maketitle

\section{Introduction}\label{sec:intro}

\subsection{Motivation and background}
The mutual information is a commonly used correlation measure in quantum information theory. It can be written in several equivalent forms, for instance,
\begin{align}\label{eq:iab}
    I(A:B)_\rho
    &=D(\rho_{AB}\| \rho_A\otimes \rho_B)
    =\inf_{\tau_B} D(\rho_{AB}\| \rho_A\otimes \tau_B)
    =\inf_{\sigma_A,\tau_B} D(\rho_{AB}\| \sigma_A\otimes \tau_B),
\end{align}
where $\rho_{AB}$ is a bipartite quantum state, $\rho_A$ and $\rho_B$ are its marginals on $A$ and $B$, $D$ denotes the relative entropy, and the minimizations are over quantum states $\sigma_A$ and $\tau_B$. 
These formulations are equivalent because the infima are attained at $\sigma_A=\rho_A$ and $\tau_B=\rho_B$~\cite{hayashi_2016-1,Gupta_2014}. 
Corresponding R\'enyi generalizations of the expressions in~\eqref{eq:iab} have also been studied. 
In particular, based on the Petz divergence $D_\alpha$, the \emph{non-minimized}, \emph{singly minimized}, and \emph{doubly minimized Petz R\'enyi mutual information (PRMI)} have been introduced as
\begin{align}
    I_\alpha^{\uparrow\uparrow}(A:B)_\rho 
    &\coloneqq D_\alpha (\rho_{AB}\| \rho_A\otimes \rho_B),
    \\
    I_\alpha^{\uparrow\downarrow}(A:B)_\rho 
    &\coloneqq \inf_{\tau_B}D_\alpha (\rho_{AB}\| \rho_A\otimes \tau_B),
    \\
    I_\alpha^{\downarrow\downarrow}(A:B)_\rho 
    &\coloneqq \inf_{\sigma_A,\tau_B}D_\alpha (\rho_{AB}\| \sigma_A\otimes \tau_B),
\end{align}
respectively, with applications in hypothesis testing~\cite{hayashi_2016-1,burri2025doubly}. 
Specifically, previous work has established that the non-minimized, singly minimized, and doubly minimized PRMI each admit an operational interpretation in terms of the \emph{direct exponent} of certain binary quantum state discrimination problems, as summarized in Table~\ref{tab:direct_prmi}. 
The direct exponent characterizes the rate at which the minimal type-I error decays when the type-II error is required to vanish exponentially fast.

\begin{table}
\begin{tabular}{lp{0.28\textwidth}p{0.52\textwidth}} \toprule
\textbf{Reference}&\textbf{Null hypothesis,\newline alternative hypothesis} & \textbf{Direct exponent}\\ \toprule
\cite{Hayashi2007,Nagaoka2006,Audenaert2008}&
$H_0^n=\{\rho_{AB}^{\otimes n}\}$ \newline
$H_1^n=\{\rho_{A}^{\otimes n}\otimes \rho_{B}^{\otimes n}\}$&
For any $R\in (0,\infty)$ holds \newline
$\lim\limits_{n\rightarrow\infty}-\frac{1}{n}\log \hat{\alpha}_n(e^{-nR})
= \sup\limits_{s\in (0,1)}\frac{1-s}{s}(I_s^{\uparrow\uparrow}(A:B)_\rho -R)$.
\\ \hline
\cite{hayashi_2016-1}&
$H_0^n=\{\rho_{AB}^{\otimes n}\}$ \newline
$H_1^n=\{\rho_{A}^{\otimes n}\otimes \tau_{B}^{\otimes n}\}_{\tau_{B} }$&
For any $R\in (0,\infty)$ holds \newline
$\lim\limits_{n\rightarrow\infty}-\frac{1}{n}\log \hat{\alpha}_n(e^{-nR})
= \sup\limits_{s\in (0,1)}\frac{1-s}{s}(I_s^{\uparrow\downarrow}(A:B)_\rho -R)$.
\\ \hline
\cite{burri2025doubly} &
$H_0^n=\{\rho_{AB}^{\otimes n}\}$ \newline
$H_1^n=\{\sigma_{A}^{\otimes n}\otimes \tau_{B}^{\otimes n}\}_{\sigma_{A},\tau_{B}}$ &
For any $R\in (R_{1/2},\infty)$ holds \newline
$\lim\limits_{n\rightarrow\infty}-\frac{1}{n}\log \hat{\alpha}_n(e^{-nR})
= \sup\limits_{s\in (\frac{1}{2},1)}\frac{1-s}{s}(I_s^{\downarrow\downarrow}(A:B)_\rho -R)$.
\\ \bottomrule
\end{tabular}
\caption{Table and caption adapted from \cite{burri2025doubly}. Overview of direct exponents of certain binary quantum state discrimination problems.
Let $\rho_{AB}$ be a bipartite quantum state.
Each row pertains to a sequence of binary quantum state discrimination problems with null hypothesis $H_0^n$ and alternative hypothesis $H_1^n$ for $n\in \mathbb{N}_{>0}$.
In all three rows, the $n$th null hypothesis is given by $\rho_{AB}^{\otimes n}$. 
In the second and third row, the alternative hypotheses range over the set of quantum states $\sigma_A,\tau_B$. 
$\hat{\alpha}_n(e^{-nR})$ denotes the minimal type-I error when the type-II error is at most $e^{-nR}$. 
The papers cited in the first column derive single-letter formulas for the corresponding direct exponents, which are stated in the last column. 
The lower bound on $R$ in the third row is defined as 
$R_{1/2}\coloneqq I_{1/2}^{\downarrow\downarrow}(A:B)_\rho-\frac{1}{4}\frac{\partial}{\partial s^+} I_{s}^{\downarrow\downarrow}(A:B)_\rho\big|_{s=1/2}$.}
\label{tab:direct_prmi}
\end{table}

According to the first expression in~\eqref{eq:iab}, the mutual information is the relative entropy of the joint quantum state $\rho_{AB}$ relative to the tensor product of its marginals. 
Swapping the arguments of the relative entropy in this expression yields the \emph{lautum information}~\cite{Lautum_08,Filippo25b}, defined as
\begin{align}\label{eq:lab}
    L(A:B)_\rho\coloneqq D(\rho_A\otimes \rho_B\|\rho_{AB}).
\end{align}

On a related front, the \emph{umlaut information} was recently studied in~\cite{Filippo25,Filippo25b}, corresponding to a reversal of the arguments of the relative entropy in the second expression in~\eqref{eq:iab}. 
That is, the umlaut information is defined as
\begin{align}\label{eq:uab}
    U(A:B)_\rho \coloneqq \inf_{\tau_B}D(\rho_A\otimes \tau_B\| \rho_{AB}),
\end{align}
and has been shown to admit operational interpretations in quantum hypothesis testing and zero-rate channel coding~\cite{Filippo25b}. 

Against this background, it is natural to ask whether reversing the arguments of the relative entropy in the third expression in~\eqref{eq:iab} also yields an operationally meaningful correlation measure. 
Motivated by this question, we introduce the \emph{tumula information}, defined as 
\begin{align}\label{eq:tab}
    T(A:B)_\rho \coloneqq \inf_{\sigma_A,\tau_B}D(\sigma_A\otimes \tau_B\| \rho_{AB}).
\end{align}
This quantity combines the idea of double minimization, as in the doubly minimized PRMI, with the reversed relative entropy viewpoint of the lautum information.

Similarly, we define the \emph{non-minimized}, \emph{singly minimized}, and \emph{doubly minimized Petz R\'enyi lautum information (PRLI)} as
\begin{align}
    L_\alpha^{\uparrow\uparrow}(A:B)_\rho 
    &\coloneqq D_\alpha(\rho_A\otimes \rho_B\| \rho_{AB}),
    \\
    L_\alpha^{\uparrow\downarrow}(A:B)_\rho 
    &\coloneqq \inf_{\tau_B}D_\alpha(\rho_A\otimes \tau_B\| \rho_{AB}),
    \\
    L_\alpha^{\downarrow\downarrow}(A:B)_\rho 
    &\coloneqq \inf_{\sigma_A,\tau_B}D_\alpha(\sigma_A\otimes \tau_B\| \rho_{AB}).
\end{align}
The singly minimized PRLI was previously introduced in~\cite{Filippo25b} (where it was called Petz--R\'enyi umlaut information), while the doubly minimized PRLI is newly introduced in this work. 
With regard to these three types of PRLI, we again ask whether they possess operational significance. 
Since the corresponding types of PRMI admit an operational interpretation in the context of hypothesis testing (see Table~\ref{tab:direct_prmi}), it is natural to seek an analogous operational interpretation for the three types of PRLI. 
Indeed, for the non-minimized PRLI such an interpretation arises directly by considering the same hypothesis testing problem as for the non-minimized PRMI (see the first row of Table~\ref{tab:direct_prmi}), but with the roles of type-I and type-II errors reversed. 
Accordingly, we define the \emph{reverse direct exponent} as the rate at which the minimal type-II error decays when the type-I error is required to vanish exponentially fast. 
Since the hypothesis testing problem in the non-minimized setting corresponds to a standard i.i.d. hypothesis testing scenario, well-established results from the fundamental literature apply, implying that the reverse direct exponent is characterized by the non-minimized PRLI, as described in the first row of Table~\ref{tab:direct_prli}. 
This observation naturally leads to the question of whether analogous results hold for the singly and doubly minimized PRLI when we consider the same settings as for the singly and doubly minimized PRMI (see the second and third row of Table~\ref{tab:direct_prmi}), but evaluate the reverse direct exponent instead of the direct exponent. 
Furthermore, for the hypothesis testing problem associated with the singly minimized PRMI, previous work~\cite{Filippo25b} has shown that the Sanov exponent equals the umlaut information. 
Similarly, the question arises as to whether the Sanov exponent for the doubly minimized setting is given by the tumula information.\\

\begin{table}
\begin{tabular}{l p{0.28\textwidth} p{0.52\textwidth}} \toprule
\textbf{Reference}&\textbf{Null hypothesis,\newline alternative hypothesis} & \textbf{Reverse direct exponent}\\ \toprule
\cite{Hayashi2007,Nagaoka2006,Audenaert2008}&
$H_0^n=\{\rho_{AB}^{\otimes n}\}$ \newline 
$H_1^n=\{\rho_{A}^{\otimes n}\otimes \rho_{B}^{\otimes n}\}$ &
For any $R\in (0,\infty)$ holds \newline
$\lim\limits_{n\rightarrow\infty}-\frac{1}{n}\log \hat{\beta}_n(e^{-nR})
= \sup\limits_{s\in (0,1)}\frac{1-s}{s}(L_s^{\uparrow\uparrow}(A:B)_\rho -R)$.
\\ \hline
Theorem~\ref{thm:singly} &
$H_0^n=\{\rho_{AB}^{\otimes n}\}$ \newline
$H_1^n=\{\rho_{A}^{\otimes n}\otimes \tau_{B}^{\otimes n}\}_{\tau_{B} }$&
For any $R\in (0,\infty)$ holds \newline
$\lim\limits_{n\rightarrow\infty}-\frac{1}{n}\log \hat{\beta}_n(e^{-nR})
= \sup\limits_{s\in (0,1)}\frac{1-s}{s}(L_s^{\uparrow\downarrow}(A:B)_\rho -R)$.
\\ \hline
Theorem~\ref{thm:doubly} &
$H_0^n=\{\rho_{AB}^{\otimes n}\}$ \newline
$H_1^n=\{\sigma_{A}^{\otimes n}\otimes \tau_{B}^{\otimes n}\}_{\sigma_{A},\tau_{B} }$ &
For any $R\in (0,R^L_{1/2})\cup (T(A:B)_\rho,\infty)$ holds \newline
$\lim\limits_{n\rightarrow\infty}-\frac{1}{n}\log \hat{\beta}_n(e^{-nR})
=\sup\limits_{s\in (0,\frac{1}{2})}\frac{1-s}{s}(L_s^{\downarrow\downarrow}(A:B)_\rho -R)$.
\\ \bottomrule
\end{tabular}
\caption{Overview of reverse direct exponents of the same binary quantum state discrimination problems as in Table~\ref{tab:direct_prmi}. 
$\hat{\beta}_n(e^{-nR})$ denotes the minimal type-II error when the type-I error is at most $e^{-nR}$. 
The papers cited in the first column derive single-letter formulas for the corresponding reverse direct exponents, which are stated in the last column. 
The upper bound on $R$ in the third row is defined as 
$R_{1/2}^L\coloneqq L_{1/2}^{\downarrow\downarrow}(A:B)_\rho-\frac{1}{4}\frac{\partial}{\partial s^-} L_{s}^{\downarrow\downarrow}(A:B)_\rho\big|_{s=1/2}$.}
\label{tab:direct_prli}
\end{table}

The definition of the tumula information for states can be extended to an information measure for channels. 
For a quantum channel $\pazocal{N}$ from $A$ to $B$, the \emph{quantum channel tumula information of $\pazocal{N}$} is defined by optimizing the corresponding state quantity over all pure input states, i.e., 
\bb
T(\pazocal{N}) \coloneqq \sup_{\Psi_{A'A}} T(A':B)_{(\mathrm{Id}\otimes\pazocal{N})(\Psi)} 
\ee
where $A'$ is an auxiliary Hilbert space isomorphic to $A$, and the supremum is over all pure states $\Psi_{A'A}$. 
Motivated by the operational relevance of the channel umlaut information in the setting of non-signalling--assisted communication at low rates \cite{Filippo25,Filippo25b} (see Figure~\ref{fig:recap}), we study the properties of the channel tumula information in order to investigate a possible operational interpretation in assisted communication. 
In particular, since -- by definition -- the tumula information of a channel is smaller than its umlaut information, it is natural to wonder whether it admits an operational interpretation in a communication scenario assisted by a weaker resource than non-signalling assistance, such as entanglement assistance.

\begin{figure}[t]
    \centering
\begin{tikzpicture}
    \draw[thick, ->] (-1,0) -- (12,0);

    \draw[thick] (3,0.2) -- (3,-0.2); 
    \draw[thick] (6.5,0.2) -- (6.5,-0.2); 
    \draw[thick] (9.5,0.2) -- (9.5,-0.2); 
    \draw[thick, ForestGreen] (6.5,1.3) -- (6.5,0.9);
    \draw[thick, ForestGreen] (-1,1.1) -- (6.5,1.1);

    \node[above] at (6.5,0.2) {$U(\pazocal{N})$};
    \node[above] at (8,0.2) {$\leq$};
    \node[above] at (9.5,0.2) {$U^\infty(\pazocal{N})$};
    \node[below] at (9.5,-0.2) {$E^{\rm NS,a}(0^+,\pazocal{N})$};
    \node[below] at (3,-0.2) {$E^{\emptyset}(0^+,\pazocal{N})$};
    \node[above, ForestGreen] at (2.75,1.1) {$T(\pazocal{N})$};
    
    \node[above] at (13,0.8) {information};
    \node[above] at (13,0.4) {measures};
    \node[above] at (13,-0.6) {error};
    \node[above] at (13,-1.1) {exponents};

\end{tikzpicture}
    \caption{A pictorial representation of the channel tumula information $T(\pazocal{N})$, 
    compared with the channel umlaut information $U(\pazocal{N})$, its regularization $U^\infty(\pazocal{N})$, and two relevant error exponents: the unassisted, zero-rate reliability function $E^{\emptyset}(0^+,\pazocal{N})$, and the activated, non-signalling assisted reliability function $E^{\rm NS,a}(0^+,\pazocal{N})$. 
    For the definitions of these exponents, see~\cite{Filippo25b}.}
    \label{fig:recap}
\end{figure}
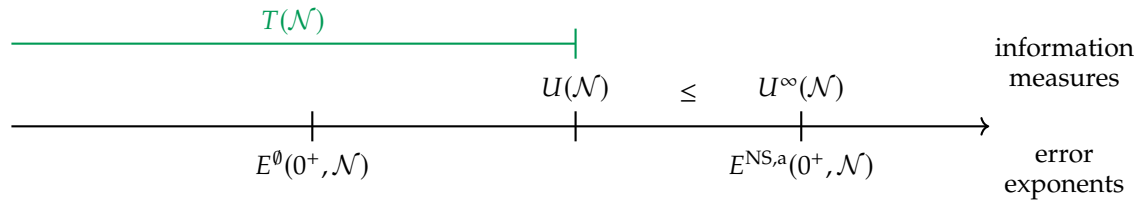

\subsection{Overview of results}

In this paper, we establish several general properties of the doubly minimized PRLI (see Theorem~\ref{thm:properties_petz}) and the tumula information (see Theorem~\ref{thm:properties}). 
For the doubly minimized PRLI of order $\alpha$, we note a direct relation to the doubly minimized PRMI of order $1-\alpha$ (see~\eqref{eq:prli_prmi2}), which allows several properties of the latter to be transferred to the former. 
In particular, this relation implies additivity of the doubly minimized PRLI of order $\alpha\in [0,1/2]$ 
and uniqueness of the minimizer for $\alpha\in (0,1/2)$. 
Moreover, it implies that the doubly minimized PRLI of order $1/2$ coincides with the doubly minimized PRMI of order $1/2$, i.e., $L_{1/2}^{\downarrow\downarrow}(A:B)_\rho=I_{1/2}^{\downarrow\downarrow}(A:B)_\rho$, where the latter is known to be equal to the min-reflected entropy~\cite{burri2025minreflectedentropydoubly}. 
For the tumula information, we prove additivity and derive an expression for $T(A:B)_\rho$ in terms of universal permutation invariant states, eliminating the need for explicit optimization and replacing it with an asymptotic limit (see~\eqref{eq:tab_omega}). 

For the same hypothesis testing problems as for the singly and doubly minimized PRMI, we study the \emph{reverse direct exponent} instead of the direct exponent. 
We show that the reverse direct exponent is determined by the singly and doubly minimized PRLI, respectively. 
These findings are summarized in the second and third row of Table~\ref{tab:direct_prli} (see also Theorem~\ref{thm:singly}, \ref{thm:doubly}). 
Qualitatively, comparing Tables~\ref{tab:direct_prmi} and~\ref{tab:direct_prli} reveals that reversing the roles of the type-I and the type-II errors ($\hat{\alpha} \leftrightarrow \hat{\beta}$) mirrors the reversal of the arguments of the Petz divergence (PRMI $\leftrightarrow$ PRLI). 
Moreover, for the hypothesis testing problem associated with the doubly minimized PRLI, we show that the \emph{Sanov exponent} equals the tumula information (see Corollary~\ref{cor:sanov}). 
These results provide operational interpretations for the singly minimized PRLI of order $\alpha\in (0,1)$, the doubly minimized PRLI of order $\alpha\in (0,1/2)$, and the tumula information. 
As an additional result, we note that taking the zero rate limits of the relations in Tables~\ref{tab:direct_prmi} and~\ref{tab:direct_prli} yields further operational interpretations of the lautum, umlaut, tumula, and the mutual information, respectively (see Corollary~\ref{cor:lowrate}).

We also explore the application of the tumula information in the context of channel coding by introducing the \emph{channel tumula information}, both in the quantum and classical setting. 
For the class of classical-quantum (CQ) channels, we derive several equivalent variational expressions and compare them with the corresponding expressions for the channel umlaut information. 
Furthermore, we investigate whether the (regularized) channel tumula information could characterize a zero-rate error exponent in an assisted communication setting, analogous to the channel umlaut information, and answer this question in the negative for the classical setting.

\emph{Outline.} 
We begin in Section~\ref{sec:preliminaries} with mathematical preliminaries, introducing our general notation 
and the relevant definitions associated with divergences (\ref{ssec:divergences}) and variants of the mutual information (\ref{ssec:variants}). 
Section~\ref{sec:properties} presents our results on elementary properties of the doubly minimized PRLI and the tumula information. 
In Section~\ref{sec:discrimination}, we discuss applications in hypothesis testing. 
Section~\ref{sec:communication} is devoted to the study of the channel tumula information.

\section{Preliminaries}\label{sec:preliminaries}

Throughout this work, we restrict attention to finite alphabets in the classical setting and to finite-dimensional Hilbert spaces in the quantum setting, for simplicity. 
A summary of frequently used notation is provided in Table~\ref{tab:notation}.

\begin{table}[h]
    \caption{Notational conventions}
    \label{tab:notation}
\begin{tabularx}{\textwidth}{p{0.13\textwidth}X}\toprule
Symbol&Description\\ \midrule
$\log (x)$&Natural logarithm of $x$\\
$[n]$& for $n\in \mathbb{N}$: the set $\{0,1,\dots, n-1\}$\\
$A$&Hilbert space on system $A$\\
$d_A$&dimension of Hilbert space $A$\\
$\mathcal{L}(A)$&Set of linear maps from $A$ to $A$\\
$\mathcal{S}(A)$
&Set of quantum states on $A$ (i.e., density matrices on $A$)\\
$\mathcal{S}_{\ll X}(A)$
&Set of quantum states $\rho$ on $A$ such that $\rho\ll X$\\
$\cptp (A,B)$
&Set of completely positive, trace-preserving linear maps from $\mathcal{L}(A)$ to $\mathcal{L}(B)$ (i.e., quantum channels from $A$ to $B$)\\
$X\ll Y$
&for $X,Y\in \mathcal{L}(A)$: kernel of $Y$ is a subset of the kernel of $X$\\
$X\perp Y$
&for $X,Y\in \mathcal{L}(A)$: $XY=YX=0$\\
$X^\dagger$
&for $X\in \mathcal{L}(A)$: Hermitian adjoint of $X$\\
$X\geq 0$
& for $X\in \mathcal{L}(A)$: $X$ is positive semidefinite\\
$X\geq Y$
& for self-adjoint $X,Y\in \mathcal{L}(A)$: $X-Y\geq 0$\\
$X^p$
&for positive semidefinite $X\in \mathcal{L}(A)$: $X$ to the power of $p$; power taken on the support of $X$\\
$|X|$
&Operator absolute value of $X$\\
$\|X \|_p$
&Schatten $p$-(quasi-)norm of $X$ for $p\in (0,\infty]$ \\
$\{X\geq Y\}$
&for self-adjoint $X,Y\in \mathcal{L}(A)$: orthogonal projection onto the subspace corresponding to the non-negative eigenvalues of $X-Y$\\
$\{X< Y\}$
&for self-adjoint $X,Y\in \mathcal{L}(A)$: orthogonal projection onto the subspace corresponding to the strictly negative eigenvalues of $X-Y$\\
$\mathrm{Sym}^n(A)$
&symmetric subspace of $A^{\otimes n}$\\
$\mathcal{L}_{\sym}(A^n)$
&set of linear operators on $A^n\equiv A^{\otimes n}$ that are permutation invariant\\
$\mathcal{S}_{\sym}(A^n)$
&set of quantum states on $A^n\equiv A^{\otimes n}$ that are permutation invariant\\
$\omega_{A^n}^n$
&universal permutation invariant state~\cite{renner2006security,christandl2009postselection,hayashi_2016-1} defined as 
$\omega_{A^n}^n \coloneqq g_{n,d_A}^{-1}\Tr_{{A'}^n}[(P^{n}_{\sym})_{A^n{A'}^n}]$ where $A'$ is a Hilbert space isomorphic to $A$, 
$g_{n,d_A}\coloneqq \dim (\mathrm{Sym}^n(A\otimes A'))$ 
and $P^{n}_{\sym}$ denotes the orthogonal projection onto $\mathrm{Sym}^n(A\otimes A')$.
\\
\bottomrule
\end{tabularx}
\end{table}

\newpage
\subsection{Divergences}\label{ssec:divergences}
The \emph{classical relative entropy (or Kullback-Leibler divergence)} of a probability distribution $P$ relative to a probability distribution $Q$ is
\begin{align}
    D(P\| Q)\coloneqq
    \sum_{x\in X}P_X(x)\log \frac{P_X(x)}{Q_X(x)}
\end{align}
if $\not\exists x\in X:(P_X(x)>0 \land Q_X(x)=0)$, and $D(P\| Q)\coloneqq\infty$ else.

The \emph{(quantum) relative entropy} of $\rho\in \mathcal{S}(A)$ relative to a positive semidefinite $\sigma\in \mathcal{L}(A)$ is
\begin{align}
D(\rho\| \sigma)\coloneqq
\Tr[\rho (\log\rho -\log\sigma)]
\end{align}
if $\rho\ll \sigma$ and $D(\rho\| \sigma)\coloneqq \infty$ else.

The \emph{von Neumann entropy} of $\rho\in \mathcal{S}(A)$ is defined as $H(A)_\rho\coloneqq -\Tr[\rho\log \rho]$.
For $\rho\in \mathcal{S}(AB)$, 
the \emph{conditional entropy} of $A$ given $B$ is $H(A|B)_\rho \coloneqq H(AB)_\rho -H(B)_\rho$ and
the \emph{mutual information} between $A$ and $B$ is $I(A:B)_\rho \coloneqq H(A)_\rho+H(B)_\rho-H(AB)_\rho$.
The \emph{R\'enyi entropy (of order $\alpha$)} of $\rho\in \mathcal{S}(A)$ is defined as
\bb
H_\alpha (A)_\rho
\coloneqq\frac{1}{1-\alpha}\log \Tr[\rho^\alpha]
\ee
for $\alpha\in (-\infty,1)\cup (1,\infty)$, 
and for $\alpha \in \{1,\infty\}$ as the corresponding limits. 

The mutual information of $\rho_{AB}\in \mathcal{S}(AB)$ can be expressed in terms of the relative entropy in the following ways~\cite{hayashi_2016-1,Gupta_2014}.
\begin{equation}\label{eq:i1}
I(A:B)_\rho
=D (\rho_{AB}\| \rho_A\otimes \rho_B)
=\inf_{\tau_B\in \mathcal{S}(B)}D (\rho_{AB}\| \rho_A\otimes \tau_B)
=\inf_{\substack{\sigma_A\in \mathcal{S}(A),\\ \tau_B\in \mathcal{S}(B)}}D (\rho_{AB}\| \sigma_A\otimes \tau_B)
\end{equation}

The \emph{Petz (quantum R\'enyi) divergence} of order $\alpha\in (0,1)\cup (1,\infty)$ of $\rho\in \mathcal{S}(A)$ relative to a positive semidefinite $\sigma\in \mathcal{L}(A)$ is given by~\cite{PetzRenyi}
\begin{equation}
D_\alpha (\rho\| \sigma)\coloneqq\frac{1}{\alpha -1} \log \Tr [\rho^\alpha \sigma^{1-\alpha}]
\end{equation}
if $(\alpha <1\land \rho\not\perp\sigma)\lor \rho\ll \sigma$ and 
$D_\alpha (\rho\| \sigma)\coloneqq \infty$ else. 
Moreover, $D_0$ and $D_1$ are defined as the limits of $D_\alpha$ for $\alpha\rightarrow\{0,1\}$.
For $\alpha\in (-\infty,\infty)$, we define
$Q_\alpha (\rho\| \sigma)\coloneqq
\Tr [\rho^\alpha \sigma^{1-\alpha}]$
for all positive semidefinite $\rho,\sigma\in \mathcal{L}(A)$. 
By the definition of the Petz divergence, we have for all $\alpha\in (0,1)$
\bb\label{eq:reverse}
    D_\alpha(\rho\|\sigma)=\frac{\alpha}{1-\alpha} D_{1-\alpha}(\sigma\|\rho).
\ee

The Petz divergence satisfies the following quantum Sibson identity. 
For any $\alpha\in [0,1)\cup (1,\infty),\rho_{AB}\in \mathcal{S}(AB),\sigma_A\in \mathcal{S}(A),\tau_B\in \mathcal{S}(B)$ such that $\rho_A\not\perp\sigma_A$ holds
\begin{align}\label{eq:sibson}
D_\alpha(\sigma_A\otimes \tau_B\| \rho_{AB})
=D_\alpha (\sigma_A\otimes \hat{\tau}_B\| \rho_{AB})
+D_\alpha (\tau_B\|\hat{\tau}_B)
\quad\text{where}\quad 
\hat{\tau}_B\coloneqq \frac{(\Tr_A[\rho_{AB}^{1-\alpha} \sigma_A^{\alpha}])^{\frac{1}{1-\alpha}}}{\Tr[(\Tr_A[\rho_{AB}^{1-\alpha} \sigma_A^{\alpha}])^{\frac{1}{1-\alpha}}]}.
\end{align}
This relation follows immediately from the definition of the Petz divergence. 
(The Sibson identity in~\eqref{eq:sibson} is the reversed version of the Sibson identity in~\cite{hayashi_2016-1,burri2025alternatingminimizationcomputingdoubly}.)

\subsection{Mutual information and its variants}\label{ssec:variants}

In this section, we restate several definitions mentioned in the introduction. 

Given a bipartite state $\rho_{AB}\in\mathcal{S}(AB)$, the \emph{mutual information} is defined by any of the following equivalent expressions
\begin{align}
    I(A:B)_\rho\coloneqq D(\rho_{AB}\|\rho_A\otimes \rho_B)
    =\inf_{\tau_B\in \mathcal{S}(B)}D(\rho_{AB}\|\rho_A\otimes\tau_B)
    =\inf_{\substack{\sigma_A\in \mathcal{S}(A),\\ \tau_B\in \mathcal{S}(B)}}D(\rho_{AB}\|\sigma_A\otimes\tau_B).
\end{align}
The \emph{(quantum) lautum, umlaut,} and \emph{tumula information} are defined as
\begin{align}
L(A:B)_{\rho}&\coloneqq D(\rho_A\otimes\rho_B\|\rho_{AB}),
\\
U(A:B)_{\rho}&\coloneqq \inf_{\tau_B\in \mathcal{S}(B)}D(\rho_A\otimes\tau_B\|\rho_{AB}),
\\
T(A:B)_{\rho}&\coloneqq \inf_{\substack{\sigma_A\in \mathcal{S}(A),\\ \tau_B\in \mathcal{S}(B)}}D(\sigma_A\otimes\tau_B\|\rho_{AB}).
\label{eq:def_umlaut}
\end{align}

We define the \emph{non-minimized, singly minimized,} and \emph{doubly minimized Petz R\'enyi mutual information (PRMI)} to be
\begin{align}
I_\alpha^{\uparrow\uparrow}(A:B)_\rho 
&\coloneqq D_\alpha (\rho_{AB}\|\rho_A\otimes \rho_B),
\\
I_\alpha^{\uparrow\downarrow}(A:B)_\rho 
&\coloneqq \inf_{\tau_B\in \mathcal{S}(B)} D_\alpha (\rho_{AB}\|\rho_A\otimes \tau_B),
\\
I_\alpha^{\downarrow\downarrow}(A:B)_\rho 
&\coloneqq \inf_{\substack{\sigma_A\in \mathcal{S}(A), \\ \tau_B\in \mathcal{S}(B)}} D_\alpha (\rho_{AB}\|\sigma_A\otimes \tau_B),
\end{align}
respectively. For $\alpha=1$ and $\alpha\rightarrow 1$, all of them coincide with the mutual information~\cite{Gupta_2014,hayashi_2016-1,burri2025doubly}, i.e.,
\begin{align}
    I(A:B)_\rho 
    &=I_1^{\uparrow\uparrow}(A:B)_\rho
    =\lim_{\alpha\rightarrow 1}I_\alpha^{\uparrow\uparrow}(A:B)_\rho
    \\
    &=I_1^{\uparrow\downarrow}(A:B)_\rho
    =\lim_{\alpha\rightarrow 1} I_\alpha^{\uparrow\downarrow}(A:B)_\rho
    \\
    &=I_1^{\downarrow\downarrow}(A:B)_\rho
    =\lim_{\alpha\rightarrow 1} I_\alpha^{\downarrow\downarrow}(A:B)_\rho.
\end{align}

Furthermore, we define the \emph{non-minimized, singly minimized,} and \emph{doubly minimized Petz R\'enyi lautum information (PRLI)} to be
\begin{align}
L_\alpha^{\uparrow\uparrow}(A:B)_\rho 
&\coloneqq D_\alpha (\rho_A\otimes \rho_B\| \rho_{AB}),
\\
L_\alpha^{\uparrow\downarrow}(A:B)_\rho 
&\coloneqq \inf_{\tau_B\in \mathcal{S}(B)} D_\alpha (\rho_A\otimes \tau_B\| \rho_{AB}),
\\
L_\alpha^{\downarrow\downarrow}(A:B)_\rho 
&\coloneqq \inf_{\substack{\sigma_A\in \mathcal{S}(A), \\ \tau_B\in \mathcal{S}(B)}} D_\alpha (\sigma_A\otimes \tau_B\| \rho_{AB}),
\end{align}
respectively. 
Remarkably, these three types of PRLI correspond to the analogous types of PRMI. By~\eqref{eq:reverse}, we have for all $\alpha\in (0,1)$
\begin{align}
L_\alpha^{\uparrow\uparrow}(A:B)_\rho 
&=\frac{\alpha}{1-\alpha}I_{1-\alpha}^{\uparrow\uparrow}(A:B)_\rho,
\label{eq:prli_prmi0}\\
L_\alpha^{\uparrow\downarrow}(A:B)_\rho 
&=\frac{\alpha}{1-\alpha}I_{1-\alpha}^{\uparrow\downarrow}(A:B)_\rho,
\label{eq:prli_prmi1}\\
L_\alpha^{\downarrow\downarrow}(A:B)_\rho 
&=\frac{\alpha}{1-\alpha}I_{1-\alpha}^{\downarrow\downarrow}(A:B)_\rho.
\label{eq:prli_prmi2}
\end{align}

For $\alpha=1$ and $\alpha\rightarrow 1$, the three types of PRLI correspond to the lautum, umlaut, and tumula information, respectively, i.e.,
\begin{equation}\label{eq:prli-limits}
        \begin{aligned}
        L(A:B)_\rho &= L_1^{\uparrow\uparrow}(A:B)_\rho
        =\lim_{\alpha\to 1}L_\alpha^{\uparrow\uparrow}(A:B)_\rho,\\ 
        U(A:B)_\rho &= L_1^{\uparrow\downarrow}(A:B)_\rho
        =\lim_{\alpha\to 1}L_\alpha^{\uparrow\downarrow}(A:B)_\rho, \\ 
        T(A:B)_\rho &= L_1^{\downarrow\downarrow}(A:B)_\rho
        =\lim_{\alpha\to 1}L_\alpha^{\downarrow\downarrow}(A:B)_\rho.
        \end{aligned}
\end{equation}
The proof technique is standard in the literature. 
For completeness, we provide a concise proof of these equalities in Appendix~\ref{proof:limits}.

\begin{rem}[(Classical setting)]
    All notions of mutual information introduced in this section can be defined analogously in the classical setting. 
    For instance, given a probability distribution $P_{XY}$ of two random variables $X,Y$ over alphabets $\mathcal{X},\mathcal{Y}$, the \emph{classical lautum}, \emph{umlaut}, and \emph{tumula information} are defined as
    \begin{align}
        L(X:Y)_P&\coloneqq D(P_XP_Y\| P_{XY}),\\
        U(X:Y)_P&\coloneqq \inf_{R_Y}D(P_XR_Y\| P_{XY}),\\
        T(X:Y)_P&\coloneqq \inf_{Q_X,R_Y}D(Q_XR_Y\| P_{XY}),
    \end{align}
    where the minimizations are over probability distributions $Q_X$ and $R_Y$ of $X$ and $Y$, respectively.
\end{rem}

\section{Properties}\label{sec:properties}

The following theorem presents several properties of the doubly minimized PRLI.

\begin{thm}[(Properties of doubly minimized PRLI)]\label{thm:properties_petz}
Let $\rho_{AB}\in \mathcal{S}(AB)$.
\begin{enumerate}[label=(\alph*)]
\item \emph{Monotonicity under local operations:} 
$L_\alpha^{\downarrow\downarrow}(A':B')_{\pazocal{N} \otimes \pazocal{M} (\rho_{AB})} \leq  L_\alpha^{\downarrow\downarrow}(A:B)_\rho$ for all $\alpha\in [0,2]$, $\pazocal{N}\in \cptp(A,A'),\pazocal{M}\in \cptp(B,B')$.
\item \emph{Non-negativity:} 
$L_\alpha^{\downarrow\downarrow}(A:B)_\rho\in [0,\infty]$ for all $\alpha\in [0,\infty)$, and it is finite for all $\alpha\in [0,1)$. 
\item \emph{Additivity:} 
$L_\alpha^{\downarrow\downarrow}(A_1A_2:B_1B_2)_{\rho_{A_1B_1}\otimes \rho'_{A_2B_2}}
=L_\alpha^{\downarrow\downarrow}(A_1:B_1)_{\rho_{A_1B_1}}
+L_\alpha^{\downarrow\downarrow}(A_2:B_2)_{\rho'_{A_2B_2}}$ for all $\alpha \in [0,\frac{1}{2}]$.

\item \emph{Joint convexity:} For any $\alpha\in (0,\frac{1}{2}],\lambda\in [0,1],\lambda'\coloneqq 1-\lambda,\sigma_A\in \mathcal{S}(A),\sigma_A'\in \mathcal{S}(A),
\tau_B\in \mathcal{S}(B),\tau_B'\in \mathcal{S}(B)$
\begin{align}
Q_\alpha((\lambda \sigma_A+\lambda'\sigma_A')\otimes (\lambda \tau_B+\lambda'\tau_B')\| \rho_{AB})
&\geq \lambda Q_\alpha( \sigma_A\otimes \tau_B\| \rho_{AB}) +\lambda' Q_\alpha(\sigma_A'\otimes \tau_B'\| \rho_{AB}),
\label{eq:joint-concavity}
\\
D_\alpha((\lambda \sigma_A+\lambda'\sigma_A')\otimes (\lambda \tau_B+\lambda'\tau_B')\| \rho_{AB})
&\leq \lambda D_\alpha(\sigma_A\otimes \tau_B\| \rho_{AB}) +\lambda' D_\alpha(\sigma_A'\otimes \tau_B'\| \rho_{AB}).
\label{eq:joint-convexity}
\end{align}

\item \emph{Uniqueness of minimizer:} 
Let $\alpha\in (0,\frac{1}{2})$. 
Then there exists 
$({\sigma}_A^\star,{\tau}_B^\star)\in \mathcal{S}(A)\times \mathcal{S}(B)$ such that 
\begin{align}
\argmin_{(\sigma_A,\tau_B)\in \mathcal{S}(A)\times\mathcal{S}(B)} 
D_\alpha (\sigma_A\otimes \tau_B\| \rho_{AB})
=\{({\sigma}_A^\star,{\tau}_B^\star)\},
\end{align}
and ${\sigma}_A^\star$ has the same support as $\rho_A$, and
${\tau}_B^\star$ has the same support as $\rho_B$.

Moreover, for all $\alpha\in (0,1)$
\begin{equation}
    \argmin_{(\sigma_A,\tau_B)\in \mathcal{S}(A) \times \mathcal{S}(B)}
    D_\alpha (\sigma_A\otimes \tau_B\| \rho_{AB})
    =\argmin_{(\sigma_A,\tau_B)\in \mathcal{S}(A) \times \mathcal{S}(B)}
    D_{1-\alpha}(\rho_{AB}\| \sigma_A\otimes \tau_B).
\end{equation}
\item \emph{Asymptotic optimality of permutation invariant state:} 
For any $\alpha\in [0,1)$
\begin{align}
L_\alpha^{\downarrow\downarrow}(A:B)_\rho &=\frac{\alpha}{1-\alpha} \lim_{n\rightarrow\infty} \frac{1}{n} D_{1-\alpha}(\rho_{AB}^{\otimes n}\| \omega_{A^n}^n\otimes \omega_{B^n}^n).
\end{align}
\item \emph{Partial minimizers:} 
Let $\alpha\in (0,1)$. 
For any fixed $\sigma_A\in \mathcal{S}(A)$ whose support is not orthogonal to that of $\rho_A$, we have 
\begin{align}
    \tau_B^\star\coloneqq \frac{(\Tr_A[\sigma_A^\alpha \rho_{AB}^{1-\alpha}])^{\frac{1}{1-\alpha}}}{\Tr[(\Tr_A[\sigma_A^\alpha \rho_{AB}^{1-\alpha}])^{\frac{1}{1-\alpha}}]}
\in \argmin_{\tau_B\in \mathcal{S}(B)}D_\alpha (\sigma_A\otimes \tau_B\| \rho_{AB}),
\end{align}
and for any fixed $\tau_B\in \mathcal{S}(B)$ whose support is not orthogonal to that of $\rho_B$, we have 
\begin{align}
    \sigma_A^\star\coloneqq \frac{(\Tr_B[\tau_B^\alpha \rho_{AB}^{1-\alpha}])^{\frac{1}{1-\alpha}}}{\Tr[(\Tr_B[\tau_B^\alpha \rho_{AB}^{1-\alpha}])^{\frac{1}{1-\alpha}}]}
\in \argmin_{\sigma_A\in \mathcal{S}(A)}D_\alpha (\sigma_A\otimes \tau_B\| \rho_{AB}).
\end{align}

\item \emph{Monotonicity in $\alpha$:} 
If $\alpha,\beta\in [0,\infty)$ are such that $\alpha\leq \beta$, then $L_\alpha^{\downarrow\downarrow}(A:B)_\rho \leq L_\beta^{\downarrow\downarrow}(A:B)_\rho$. 
\item \emph{Continuity in $\alpha$:} 
The function $[0,\infty)\rightarrow[0,\infty],\alpha\mapsto L_\alpha^{\downarrow\downarrow}(A:B)_\rho$ is continuous.
\item \emph{Convexity in $\alpha$:} 
The function $[0,1)\rightarrow\mathbb{R},\alpha\mapsto (\alpha-1)L_\alpha^{\downarrow\downarrow}(A:B)_\rho$ is convex.
\item \emph{Differentiability in $\alpha$:} 
The function $(0,\frac{1}{2})\rightarrow [0,\infty),\alpha \mapsto L_\alpha^{\downarrow\downarrow}(A:B)_\rho$ is continuously differentiable. 
For any $\alpha\in (0,\frac{1}{2})$ and any fixed $(\sigma_A,\tau_B)\in \argmin_{(\sigma_A',\tau_B')\in \mathcal{S}(A)\times\mathcal{S}(B) } D_\alpha (\sigma_A'\otimes \tau_B'\| \rho_{AB})$, the derivative at $\alpha$ is 
\begin{align}
    \frac{\dd}{\dd\alpha}L_\alpha^{\downarrow\downarrow}(A:B)_\rho 
    &=\frac{\partial}{\partial\alpha}D_\alpha (\sigma_A\otimes \tau_B\| \rho_{AB})\\
    &= \frac{1}{(1-\alpha)^2}I_{1-\alpha}^{\downarrow\downarrow}(A:B)_\rho - \frac{\alpha}{1-\alpha} \frac{\partial}{\partial \beta} I_\beta^{\downarrow\downarrow}(A:B)_\rho\big|_{\beta=1-\alpha}.
    \label{eq:diff-prmi2}
\end{align}
\item \emph{Special values of $\alpha$:} 
\begin{align}
    L_0^{\downarrow\downarrow}(A:B)_\rho&=0\\
    L_{1/2}^{\downarrow\downarrow}(A:B)_\rho&=I_{1/2}^{\downarrow\downarrow}(A:B)_\rho
    \label{eq:l12-i12}\\
    L_1^{\downarrow\downarrow}(A:B)_\rho &= T(A:B)_\rho
    \label{eq:l1-t}
\end{align}
\item \emph{Product states:} 
If $\rho_{AB}=\rho_A\otimes \rho_B$, then $L_\alpha^{\downarrow\downarrow}(A:B)_\rho=0$ for all $\alpha\in [0,\infty)$. 
Conversely, for any $\alpha\in (0,\infty)$, if $L_\alpha^{\downarrow\downarrow}(A:B)_\rho=0$, then $\rho_{AB}=\rho_A\otimes \rho_B$.
\item \emph{CC states:} 
Let $P_{XY}$ be the joint probability distribution of two random variables $X,Y$ over $\mathcal{X}\coloneqq [d_A],\mathcal{Y}\coloneqq [d_B]$. 
If there exist orthonormal bases $\{\ket{a_x}_A\}_{x\in [d_A]},\{\ket{b_y}_B\}_{y\in [d_B]}$ for $A,B$ such that 
$\rho_{AB}=\sum_{x\in \mathcal{X}}\sum_{y\in \mathcal{Y}}P_{XY}(x,y)\ketbra{a_x,b_y}_{AB}$, then for all $\alpha\in [0,\infty)$
\begin{align}
L_\alpha^{\downarrow\downarrow}(A:B)_\rho 
=L_\alpha^{\downarrow\downarrow}(X:Y)_P.
\end{align}
\end{enumerate}
\end{thm}

\begin{proof}
    See Appendix~\ref{proof:properties2}.
\end{proof}

The above theorem establishes several properties of the doubly minimized PRLI. 
For $\alpha=1$, this quantity reduces, by definition, to the tumula information (see~\eqref{eq:l1-t}). 
The following theorem presents additional properties of the tumula information not already covered above.

\begin{thm}[(Properties of the tumula information)] \label{thm:properties}
Let $\rho_{AB}\in \mathcal{S}(AB)$.
\begin{enumerate}[label=(\alph*)]
\item \emph{Partial minimizers:} 
If $T(A:B)_\rho<\infty$, then 
for any fixed $\sigma_A\in \mathcal{S}(A)$ whose support is not orthogonal to that of $\rho_A$, we have 
\begin{align}\label{eq:el}
    \tau_B^\star &\coloneqq \frac{\exp[\Tr_A[\sigma_A\log\rho_{AB}]]}{\Tr[\exp[\Tr_A[\sigma_A\log\rho_{AB}]]]}
\in \argmin_{\tau_B\in \mathcal{S}(B)}D(\sigma_A\otimes \tau_B\| \rho_{AB}),
\end{align}
and for any fixed $\tau_B\in \mathcal{S}(B)$ whose support is not orthogonal to that of $\rho_B$, we have 
\begin{align}
    \sigma_A^\star &\coloneqq \frac{\exp[\Tr_B[\tau_B\log\rho_{AB}]]}{\Tr[\exp[\Tr_B[\tau_B\log\rho_{AB}]]]}
\in \argmin_{\sigma_A\in \mathcal{S}(A)}D(\sigma_A\otimes \tau_B\| \rho_{AB}).
\end{align}

\item \emph{Additivity:} 
$T(A_1A_2:B_1B_2)_{\rho_{A_1B_1}\otimes \rho'_{A_2B_2}} 
= T(A_1 : B_1)_{\rho_{A_1B_1}} + T(A_2:B_2)_{\rho'_{A_2B_2}}$
\item \emph{Pure states:} 
Suppose $\rho_{AB}\equiv \ketbra{\rho}_{AB}$ is a pure state. Then, 
$T(A:B)_{\ket{\rho}\bra{\rho}_{AB}}=0$ if $\ketbra{\rho}_{AB}$ is a product state, and $T(A:B)_{\ket{\rho}\bra{\rho}_{AB}}=\infty$ else.
\item \emph{Universal permutation invariant state:} 
We have
\begin{align}\label{eq:tab_omega}
    T(A:B)_\rho &= \lim_{n\rightarrow\infty} \frac{1}{\sqrt{n}} D_{\frac{1}{\sqrt{n}}}(\rho_{AB}^{\otimes n}\| \omega_{A^n}^n\otimes \omega_{B^n}^n).
\end{align}
\end{enumerate}
\end{thm}

\begin{proof}
    See Appendix~\ref{proof:properties}.
\end{proof}

Although the quantum tumula information can be infinite (see Theorem \ref{thm:properties}~(c)), the classical tumula information is generally finite. 
Indeed, it is upper bounded by the cardinalities of the alphabets $\mathcal{X}$ and $\mathcal{Y}$ associated with the random variables $X$ and $Y$, as shown in the following proposition.

\begin{prop}[(Upper bound on classical tumula information)]\label{prop:bound} Let $P_{XY}$ be the joint probability distribution of two random variables $X,Y$ over $\mathcal{X},\mathcal{Y}$. 
Then
\bb\label{eq:bound_T}
    T(X:Y)_P\leq \log \min\{|\mathcal{X}|,|\mathcal{Y}|\}.
\ee
Furthermore, this bound is tight, i.e., there exists a probability distribution $P_{XY}$ such that \eqref{eq:bound_T} is an equality.
\end{prop}
\begin{proof}
    See Appendix~\ref{proof:bound}.
\end{proof}

\newpage
\section{Operational interpretation in binary quantum state discrimination}\label{sec:discrimination}
\subsection{Composite asymmetric hypothesis testing}

Let $\mathcal{H}$ be a Hilbert space and let $H_0=(H_0^n)_{n\geq 1}$ and $H_1=(H_1^n)_{n\geq 1}$ be sequences of sets of quantum states such that $H_0^n,H_1^n\subseteq \mathcal{S}(\mathcal{H}^{\otimes n})$. 
In the task of hypothesis testing, one is given a sequence of states $(\xi_n)_{n\geq 1}$, 
and the task is to decide whether $\xi_n\in H_0^n$ (the null hypothesis) or $\xi_n\in H_1^n$ (the alternative hypothesis). 
For each $n\in \mathbb{N}$, the decision is based on the outcome of a POVM $\{T_n,\id-T_n\}$. 
If the outcome corresponding to $T_n$ occurs, then the null hypothesis is believed to be true, and if the outcome corresponding to $\id -T_n$ occurs, then the alternative hypothesis is believed to be true. 
In such a setting, two kinds of errors can occur:
\begin{itemize}
\item type-I error: The null hypothesis holds, but the alternative hypothesis is believed to be true.
\item type-II error: The alternative hypothesis holds, but the null hypothesis is believed to be true.
\end{itemize}
The probabilities with which these errors occur in the worst case are called 
the \emph{(worst case) type-I error probability} and the \emph{(worst case) type-II error probability}. 
They are given, respectively, by
\bb
\alpha_n(T_n)&\coloneqq \sup_{\rho_{n}\in H_0^n} \Tr[\rho_n (\id-T_n)],
\\
\beta_n(T_n)&\coloneqq \sup_{\sigma_{n}\in H_1^n}\Tr[\sigma_n T_n].
\ee

In \emph{asymmetric} hypothesis testing, one is interested in minimizing one of these probabilities while keeping the other one bounded. 
We denote the \emph{minimum} type-I error probability when the type-II error probability is upper bounded by $\varepsilon\in [0,1]$ by 
\begin{align}
\hat{\alpha}_n(\varepsilon)\coloneqq \inf_{T_n\in \mathcal{L}(\mathcal{H}^{\otimes n})}\{\alpha_n(T_n) :0\leq T_n\leq \id ,\beta_n(T_n)\leq \varepsilon\}.
\end{align}
Analogously, we denote the \emph{minimum} type-II error probability when the type-I error probability is upper bounded by $\varepsilon\in [0,1]$ by 
\begin{align}
\hat{\beta}_n(\varepsilon)\coloneqq \inf_{T_n\in \mathcal{L}(\mathcal{H}^{\otimes n})}\{\beta_n(T_n) :0\leq T_n\leq \id ,\alpha_n(T_n)\leq \varepsilon\}.
\end{align}

In the limit where $n\rightarrow\infty$, the trade-off between the type-I and type-II error probabilities can be characterized by various error exponents~\cite{Mosonyi2022}. 
In this work, we are only interested in the following exponents. 
\begin{itemize}
\item The \emph{Stein exponent} is defined as 
$\text{Stein}(H_0\|H_1)\coloneqq 
\lim\limits_{\epsilon\to 0^+}\liminf\limits_{n\to \infty}-\frac{1}{n}\log \hat{\beta}_n(\varepsilon).$ 

Similarly, we define $\text{Stein}_\varepsilon(H_0\|H_1)\coloneqq \liminf\limits_{n\to \infty}-\frac{1}{n}\log \hat{\beta}_n(\varepsilon)$ for all $\varepsilon\in (0,1)$.
\item The \emph{Sanov exponent} is defined as
$\text{Sanov}(H_0\|H_1)\coloneqq 
\lim\limits_{\epsilon\to 0^+}\liminf\limits_{n\to \infty}-\frac{1}{n}\log \hat{\alpha}_n(\varepsilon).$

Similarly, we define $\text{Sanov}_\varepsilon(H_0\|H_1)\coloneqq \liminf\limits_{n\to \infty}-\frac{1}{n}\log \hat{\alpha}_n(\varepsilon)$ for all $\varepsilon\in (0,1)$.
\item The \emph{direct exponent} with respect to $R\in [0,\infty)$ 
is defined as 
$\liminf\limits_{n\rightarrow\infty} -\frac{1}{n}\log \hat{\alpha}_n (e^{-nR})$
if this limit exists, and as $+\infty$ else. 

\item The \emph{reverse direct exponent} with respect to $R\in [0,\infty)$ 
is defined as 
$\liminf\limits_{n\rightarrow\infty} -\frac{1}{n}\log \hat{\beta}_n (e^{-nR})$
if this limit exists, and as $+\infty$ else. 
\end{itemize}

\subsection{Operational interpretation of singly minimized PRLI}
The following theorem provides an operational interpretation of the singly minimized PRLI in terms of a composite hypothesis testing problem.

\begin{boxed}{}
\begin{thm}[(Reverse direct exponent, singly minimized PRLI)]\label{thm:singly}
Let $\rho_{AB}\in \mathcal{S}(AB)$. 
Consider the null hypothesis $H_0^n\coloneqq \{\rho_{AB}^{\otimes n}\}$ and 
any of the following alternative hypotheses: 
\bb
H_1^n\coloneqq \{\rho_{A}^{\otimes n}\otimes \tau_{B}^{\otimes n}\}_{\tau_{B}\in \mathcal{S}(B)}, 
\quad
H_1^n\coloneqq \{\rho_{A}^{\otimes n}\otimes \tau_{B^n}\}_{\tau_{B^n}\in \mathcal{S}_{\sym}(B^n)}, 
\quad
H_1^n\coloneqq \{\rho_{A}^{\otimes n}\otimes \tau_{B^n}\}_{\tau_{B^n}\in \mathcal{S}(B^n)}.
\ee
Then, for any $R\in (0,\infty)$
\begin{align}
\lim_{n\rightarrow\infty}-\frac{1}{n}\log \hat{\beta}_{n}(e^{-nR})
= \sup_{s\in (0,1)}\frac{1-s}{s}(L_s^{\uparrow\downarrow}(A:B)_\rho - R).
\end{align}
\end{thm}
\end{boxed}{}

\begin{proof}
    See Appendix~\ref{proof:singly}.
\end{proof}

\subsection{Operational interpretation of doubly minimized PRLI and tumula information}

The following theorem provides an operational interpretation of the doubly minimized PRLI in terms of a composite hypothesis testing problem.

\begin{boxed}{}
\begin{thm}[(Reverse direct exponent, doubly minimized PRLI)]\label{thm:doubly}
Let $\rho_{AB}\in \mathcal{S}(AB)$ and let 
\begin{align}\label{eq:def-rl12}
    R^L_{1/2} \coloneqq L_{1/2}^{\downarrow\downarrow}(A:B)_\rho-\frac{1}{4}\frac{\partial}{\partial s^-}L_s^{\downarrow\downarrow}(A:B)_\rho \big|_{s=1/2} 
    \in [0,L_{1/2}^{\downarrow\downarrow}(A:B)_\rho].
\end{align}
Consider the null hypothesis $H_0^n\coloneqq \{\rho_{AB}^{\otimes n}\}$ and 
any of the following alternative hypotheses: 
\bb
H_1^n\coloneqq \{\sigma_{A}^{\otimes n}\otimes \tau_{B}^{\otimes n}\}_{\sigma_{A}\in \mathcal{S}(A),\tau_{B}\in \mathcal{S}(B)}, 
\quad
H_1^n\coloneqq \{\sigma_{A^n}\otimes \tau_{B^n}\}_{\sigma_{A^n}\in \mathcal{S}_{\sym}(A^n),\tau_{B^n}\in \mathcal{S}_{\sym}(B^n)}. 
\ee
Then, for any $R\in (0,R^L_{1/2})\cup (T(A:B)_\rho,\infty)$
\begin{align}
\lim_{n\rightarrow\infty}-\frac{1}{n}\log \hat{\beta}_{n}(e^{-nR})
= \sup_{s\in (0,\frac{1}{2})}\frac{1-s}{s}(L_s^{\downarrow\downarrow}(A:B)_\rho - R),
\end{align}
and the same holds when $\sup_{s\in (0,\frac{1}{2})}$ is replaced by $\sup_{s\in (0,1)}$.
\end{thm}
\end{boxed}

\begin{proof}
    See Appendix~\ref{proof:doubly}.
\end{proof}

\begin{rem}[($R_{1/2}^L$ vs. $R_{1/2}$)]\label{rem:r12l}
    Note that $R_{1/2}^L$ as defined in Theorem~\ref{thm:doubly} can be different from $R_{1/2}$ as defined in Table~\ref{tab:direct_prmi}. 
    For examples, see Appendix~\ref{app:r12}.
\end{rem}

Previous work~\cite{berta_composite,Bjelakovic2005,Mosonyi_2015,Noetzel_2014} established Stein's theorem for a composite iid null hypothesis and simple iid alternative hypothesis. 
By reversing the roles of the null and the alternative hypothesis, this corresponds to a Sanov's theorem for a simple iid null hypothesis and a composite iid alternative hypothesis. 
Applying this result to our setting, we arrive at the following proposition for the first choice of the alternative hypothesis in~\eqref{eq:sanov-h1n}. 
The same assertion can be proved for the second choice of the alternative hypothesis in~\eqref{eq:sanov-h1n}. 
This follows as a corollary of the proof of achievability of Theorem~\ref{thm:doubly}. 
Corollary~\ref{cor:sanov} provides an operational interpretation of the tumula information.

\begin{boxed}{}
\begin{cor}[(Sanov exponent)]\label{cor:sanov}
Let $\rho_{AB}\in \mathcal{S}(AB)$. 
Consider the null hypothesis $H_0^n\coloneqq \{\rho_{AB}^{\otimes n}\}$ and 
any of the following alternative hypotheses: 
\bb \label{eq:sanov-h1n}
H_1^n\coloneqq \{\sigma_{A}^{\otimes n}\otimes \tau_{B}^{\otimes n}\}_{\sigma_{A}\in \mathcal{S}(A),\tau_{B}\in \mathcal{S}(B)}, 
\quad 
H_1^n\coloneqq \{\sigma_{A^n}\otimes \tau_{B^n}\}_{\sigma_{A^n}\in \mathcal{S}_{\sym}(A^n),\tau_{B^n}\in \mathcal{S}_{\sym}(B^n)}.
\ee
Then, for any $\varepsilon\in (0,1)$
\begin{align}
{\rm Sanov}_\varepsilon(H_0\| H_1) &= T(A:B)_\rho,
\end{align}
and as a consequence, 
${\rm Sanov}(H_0\| H_1) = T(A:B)_\rho$.
\end{cor}
\end{boxed}{}

\begin{proof}
    See Appendix~\ref{proof:sanov}.
\end{proof}

\subsection{Zero rate limits}
The following corollary provides an operational interpretation of the lautum, umlaut, tumula, and the mutual information, respectively. 
Parts (a), (b), and~(c) follow from the previous results on the direct exponent outlined in Table~\ref{tab:direct_prmi}. 
Part~(d) follows from the results on the reverse direct exponent outlined in Table~\ref{tab:direct_prli}. 

\begin{cor}[(Zero rate limits)]\label{cor:lowrate}
    Let $\rho_{AB}\in \mathcal{S}(AB)$. 
    Consider the null hypothesis $H_0^n\coloneqq \{\rho_{AB}^{\otimes n}\}$ and any of the following alternative hypotheses.
    \begin{enumerate}[label=(\alph*)]
    \item Let $H_1^n\coloneqq \{\rho_A^{\otimes n}\otimes \rho_B^{\otimes n}\}$. 
    Then, 
    \begin{align}
    \lim_{R\to 0^+}\lim_{n\rightarrow\infty}-\frac{1}{n}\log \hat{\alpha}_{n}(e^{-nR})=L_1^{\uparrow\uparrow}(A:B)_\rho=L(A:B)_\rho.
\end{align}
    \item Let $H_1^n\coloneqq \{\rho_A^{\otimes n}\otimes \tau_{B}^{\otimes n}\}_{\tau_{B}\in \mathcal{S}(B)}$,  
    $H_1^n\coloneqq \{\rho_A^{\otimes n}\otimes \tau_{B^n}\}_{\tau_{B^n}\in \mathcal{S}_{\sym}(B^n)}$, or
    $H_1^n\coloneqq \{\rho_A^{\otimes n}\otimes \tau_{B^n}\}_{\tau_{B^n}\in \mathcal{S}(B^n)}$.
    Then, 
    \begin{align}
    \lim_{R\to 0^+}\lim_{n\rightarrow\infty}-\frac{1}{n}\log \hat{\alpha}_{n}(e^{-nR})=L_1^{\uparrow\downarrow}(A:B)_\rho=U(A:B)_\rho.
    \end{align}
    \item Let $H_1^n\coloneqq \{\sigma_A^{\otimes n}\otimes \tau_{B}^{\otimes n}\}_{{\sigma_A\in \mathcal{S}(A),\\ \tau_{B}\in \mathcal{S}(B)}}$ or 
    $H_1^n\coloneqq \{\sigma_{A^n}\otimes \tau_{B^n}\}_{\sigma_{A^n}\in \mathcal{S}_{\sym}(A^n),\tau_{B^n}\in \mathcal{S}_{\sym}(B^n)}$.
    Let $R_{1/2}$ be defined as in Table~\ref{tab:direct_prmi}. 
    If $R_{1/2}=0$, then
    \begin{align}
    \lim_{R\to 0^+}\lim_{n\rightarrow\infty}-\frac{1}{n}\log \hat{\alpha}_{n}(e^{-nR})=L_1^{\downarrow\downarrow}(A:B)_\rho=T(A:B)_\rho.
    \end{align}
    \item Let $H_1^n$ be given by any of the expressions in (a) or (b). Then, 
\begin{align}
    \lim_{R\to 0^+}\lim_{n\rightarrow\infty}-\frac{1}{n}\log \hat{\beta}_{n}(e^{-nR})=I(A:B)_\rho.
\end{align}
    Moreover, the same holds if $H_1^n$ is given by the expression in (c) if $R_{1/2}^L>0$, where $R_{1/2}^L$ is defined as in~\eqref{eq:def-rl12}.
\end{enumerate}
\end{cor}

\begin{proof}
    See Appendix~\ref{proof:lowrate}.
\end{proof}

\section{Tumula information of channels}\label{sec:communication}

The aim of this section is the introduction and discussion of the tumula information as an information measure for quantum channels, in the same spirit as for the mutual information (which yields the capacity of a channel), the lautum information \cite{Lautum_08}, and the umlaut information \cite{Filippo25,Filippo25b}.

\subsection{Quantum channel tumula information}

Reversed mutual-information variants extend naturally from quantum states to quantum channels. 
For a quantum channel $\pazocal{N}$ from $A$ to $B$, the \emph{quantum channel umlaut information} is defined as~\cite{Filippo25b}
\begin{equation}
U(\pazocal{N})\coloneqq 
\sup_{\Psi_{A'A}} 
U(A':B)_{(\mathrm{Id}\otimes\pazocal{N})(\Psi_{A'A})},
\end{equation}
where $A'$ is a Hilbert space isomorphic to $A$, 
and the supremum is over all pure states $\Psi_{A'A}$. 

Motivated by this, we define the \emph{quantum channel tumula information} as the analogous optimization of the tumula information instead of the umlaut information.
\begin{Def}[(Quantum channel tumula information)]\label{def:channel_lautum}
Let $\pazocal{N} \in \cptp(A,B)$. 
Then the \emph{quantum channel tumula information of $\pazocal{N}$} is defined as
\bb
T(\pazocal{N}) \coloneqq \sup_{\Psi_{A'A}} 
T(A':B)_{(\mathrm{Id}\otimes\pazocal{N})(\Psi_{A'A})} 
= \sup_{\rho_{A'}\in \mathcal{S}(A')} 
\inf_{\sigma_{A'}\in \mathcal{S}(A')}\inf_{\tau_B\in \mathcal{S}(B)} 
D\left(\sigma_{A'}\otimes \tau_B \,\middle\|\, \rho_{A'}^{1/2} J_{A'B}^{(\pazocal{N})} \rho_{A'}^{1/2}\right),
\label{channel_umlaut}
\ee
where $A'$ is a Hilbert space isomorphic to $A$, 
$\Psi_{A'A}$ ranges over all pure states, 
and $J_{A'B}^{(\pazocal{N})}$ is the (unnormalized) Choi--Jamio\l kowski matrix of $\pazocal{N}$, defined as
\bb
J^{(\pazocal{N})}_{A'B}\coloneq (\mathrm{Id}_{A'}\otimes\pazocal{N}_{A\to B})(\Phi_{A'A}),
\label{unnormalized_Choi}
\ee
where $\Phi_{A'A}\coloneqq\sum_{i,j\in [d_A]}\ketbraa{ii}{jj}_{A'A}$ is the (unnormalized) maximally entangled state between $A'$ and $A$.
\end{Def}

The following proposition asserts that the quantum channel tumula information is super-additive under the tensor product of quantum channels. 
This follows immediately from the additivity of the tumula information for quantum states.
\begin{prop}[(Super-additivity)] \label{prop:superadditivity}
Let $\pazocal{N}_1$ and $\pazocal{N}_2$ be quantum channels. Then
\bb
    T(\pazocal{N}_1\otimes\pazocal{N}_2)\geq T(\pazocal{N}_1)+T(\pazocal{N}_2).
\ee
\end{prop}
\begin{proof}
    See Appendix~\ref{proof:superadditivity}.
\end{proof}

We will denote by $T^\infty(\pazocal{N})$ the regularised tumula information of the channel $\pazocal{N}$, defined as
\bb\label{eq:regularised}
    T^\infty(\pazocal{N}) \coloneqq \lim_{n\to \infty} \frac{1}{n}T(\pazocal{N}^{\otimes n}).
\ee
Due to the super-additivity of the channel tumula information, by Fekete's lemma this limit exists and can be equivalently written as
\bb
    T^\infty(\pazocal{N}) = \sup_{n\geq 1} \frac{1}{n}T(\pazocal{N}^{\otimes n}).
\ee

\subsection{The case of classical-quantum channels}

In this section, we consider the evaluation of the quantum channel tumula information for the special class of quantum channels called classical-quantum (CQ) channels. 
A \emph{CQ channel} is a quantum channel $\pazocal{N}\in \cptp(A,B)$ of the form 
$\pazocal{N}(\cdot)\coloneqq \sum_{x\in \mathcal{X}}\bra{x}\cdot \ket{x} \rho_x$ where $\mathcal{X}$ is an arbitrary finite set, 
$\{\ket{x}_A\}_{x\in \mathcal{X}}$ is an orthonormal basis for $A$, 
and $\rho_x\in \mathcal{S}(B)$ for all $x\in \mathcal{X}$.

For CQ channels, the quantum channel umlaut information can be expressed as~\cite[Proposition~19]{Filippo25b}
\begin{equation}\label{eq:u(n)_1}
    U(\pazocal{N}) = \sup_{P_X} (-\log Z(P_X))
    \quad \text{where}\quad 
    Z(P_X) \coloneqq \Tr\exp\Big(\sum_{x\in \mathcal{X}} P_X(x)\log \rho_x\Big),
\end{equation}
and the supremum is over all probability distributions $P_X$.

The following theorem establishes a similar expression for the quantum channel tumula information of CQ channels. 
\begin{thm}[(Quantum channel tumula information of CQ-channels)]\label{thm:CQ}
    Let $\pazocal{N}(\cdot) = \sum_{x\in \mathcal{X}}\bra{x}\cdot\ket{x}\rho_x$ be a CQ channel. Then the quantum channel tumula information of $\pazocal{N}$ can be expressed as
    \begin{align}\label{eq:alternative_form_channel}
        T(\pazocal{N}) = \sup_{P_X} \min_{Q_X} \big( D(Q_X\|P_X) - \log Z(Q_X)\big) 
        \quad \text{where}\quad 
         Z(Q_X) \coloneqq \Tr\exp\left(\sum_{x\in \mathcal{X}} Q_X(x)\log \rho_x\right),
    \end{align}
    and the optimizations are over probability distributions $P_X,Q_X$.
\end{thm}
\begin{proof}
    See Appendix~\ref{proof:CQ}.
\end{proof}

\begin{rem}[(Relation between~\eqref{eq:u(n)_1} and~\eqref{eq:alternative_form_channel})]
Note that if one sets $Q_X=P_X$ as an ansatz for the minimization over $Q_X$ in~\eqref{eq:alternative_form_channel}, one gets 
\begin{align}
T(\pazocal{N}) \stackrel{\eqref{eq:alternative_form_channel}}{\leq} {\displaystyle\sup_{P_X}} - \log\Tr\exp\Big(\sum_{x\in \mathcal{X}} P_X(x)\log \rho_x\Big) \eqt{\eqref{eq:u(n)_1}} U(\pazocal{N}).
\end{align}
\end{rem}

\begin{prop}[(Alternative expressions for channel umlaut and tumula information)]\label{prop:TU_CQ}
    Let \mbox{$\pazocal{N}(\cdot) = \sum_{x\in \mathcal{X}}\bra{x}\cdot\ket{x}\rho_x$} be a classical-to-quantum channel. Then
    \begin{align}\label{eq:optim}
    U(\pazocal{N}) &= \sup_{P_X} \min_\sigma \left(\sum_{x\in \mathcal{X}} P_X(x)  D(\sigma\|\rho_x) \right),
    \\
    T(\pazocal{N}) &= \sup_{P_X} \min_\sigma \left(-\log \sum_{x\in \mathcal{X}} P_X(x) e^{-D(\sigma\|\rho_x)} \right),
    \end{align}
    where the optimizations are over probability distributions $P_X$ and quantum states $\sigma\in \mathcal{S}(B)$.
\end{prop}
\begin{proof}
    See Appendix~\ref{proof:TU_CQ}.
\end{proof}

\subsection{Channel tumula information in the fully classical setting}
In this section, we consider classical channels from $\mathcal{X}$ to $\mathcal{Y}$ denoted by $\pazocal{W}_{Y|X}(y|x)\equiv \pazocal{W}(y|x)$ for all $x\in \mathcal{X},y\in \mathcal{Y}$. 
For a classical channel $\pazocal{W}$, the \emph{classical channel umlaut information} is defined as~\cite{Filippo25}
\begin{align}
    U(\pazocal{W})
    \coloneqq \sup_{P_X}U(X:Y)_{\pazocal{W}_{Y|X}P_X}\, ,
\end{align}
where the optimization is over probability distributions $P_X$ on $\mathcal{X}$. 

In direct analogy, we define the classical channel tumula information as follows.

\begin{Def}[(Classical channel tumula information)]\label{def:channel_classical}
Let $\pazocal{W}$ be a stochastic matrix from $\mathcal{X}$ to $\mathcal{Y}$ and 
let $X$ and $Y$ be random variables taking values in $\mathcal{X}$ and $\mathcal{Y}$. 
Then we define the \emph{classical channel tumula information of $\pazocal{W}$} as
\bb
T(\pazocal{W})
\coloneqq \sup_{P_X} T(X:Y)_{\pazocal{W}_{Y|X}P_X}
=\sup_{P_X}\inf_{Q_X,R_Y}D(Q_XR_Y\| \pazocal{W}_{Y|X}P_X)\, ,
\ee
where the optimizations are over probability distributions $P_X,Q_X,R_Y$.
\end{Def}

As in the quantum case, the additivity of the tumula information for states directly implies super-additivity of the channel tumula information.

\begin{prop}[(Super-additivity)] \label{prop:superadditivity_classical}
Let $\pazocal{W}_1$ and $\pazocal{W}_2$ be classical channels. Then
\bb
    T(\pazocal{W}_1\times\pazocal{W}_2)\geq T(\pazocal{W}_1)+T(\pazocal{W}_2).
\ee
\end{prop}
\begin{proof}
    See Appendix~\ref{proof:superadditivity_classical}.
\end{proof}

For classical channels, the channel umlaut information is given by~\cite[Proposition~18]{Filippo25}
\begin{equation}\label{eq:u(n)_2}
    U(\pazocal{W}) = \sup_{P_X} (-\log Z(P_X)) 
    \quad \text{where}\quad 
    Z(P_X) \coloneqq  \sum_{y\in \mathcal{Y}} e^{\sum_{x\in \mathcal{X}} P_X(x) \log \pazocal{W}(y|x)}
\end{equation}
while the classical channel tumula information has the following expression.
\begin{cor}[(Tumula information of classical channels)] Let $\pazocal{W}$ be a classical channel. Then
\begin{align}\label{eq:classical1}
    T(\pazocal{W}) &= \sup_{P_X} \min_{Q_X} \left( D(Q_X\|P_X) - \log Z(Q_X)\right) 
    \quad \text{where}\quad 
    Z(Q_X) \coloneqq  \sum_{y\in \mathcal{Y}} e^{\sum_{x\in \mathcal{X}} Q_X(x) \log \pazocal{W}(y|x)}
\end{align}
and the optimizations are over probability distributions $P_X,Q_X$. In addition, we have
\begin{align}\label{eq:classical2}
    T(\pazocal{W}) &= \sup_{P_X} \min_{R_Y} \left(-\log \sum_{x\in \mathcal{X}} P_X(x) e^{-D\left(R_Y\|\pazocal{W}(\cdot|x)\right)} \right),
\end{align}
where the optimizations are over probability distributions $P_X,R_Y$.
\end{cor}
\begin{proof}
    The proof of~\eqref{eq:classical1} is analogous to the proof of Theorem~\ref{thm:CQ}. 
    The proof of~\eqref{eq:classical2} is analogous to the proof of Proposition~\ref{prop:TU_CQ}.
\end{proof}

\begin{rem}[(Binary symmetric channel)]
The example of the binary symmetric channels illustrates that the classical channel tumula information can become smaller than the unassisted error exponent, see Figure~\ref{fig:BSC}. Therefore, the tumula information of a channel cannot generally be interpreted as an assisted error exponent. However, since the classical channel tumula information is super-additive (Proposition~\ref{prop:superadditivity_classical}), it is possible that the regularized quantity could still have an operational meaning. However, as we will see now, even the regularized channel tumula information of the identity channel is finite.
\end{rem}

\begin{figure}
        \centering
        \includegraphics[width=0.9\linewidth]{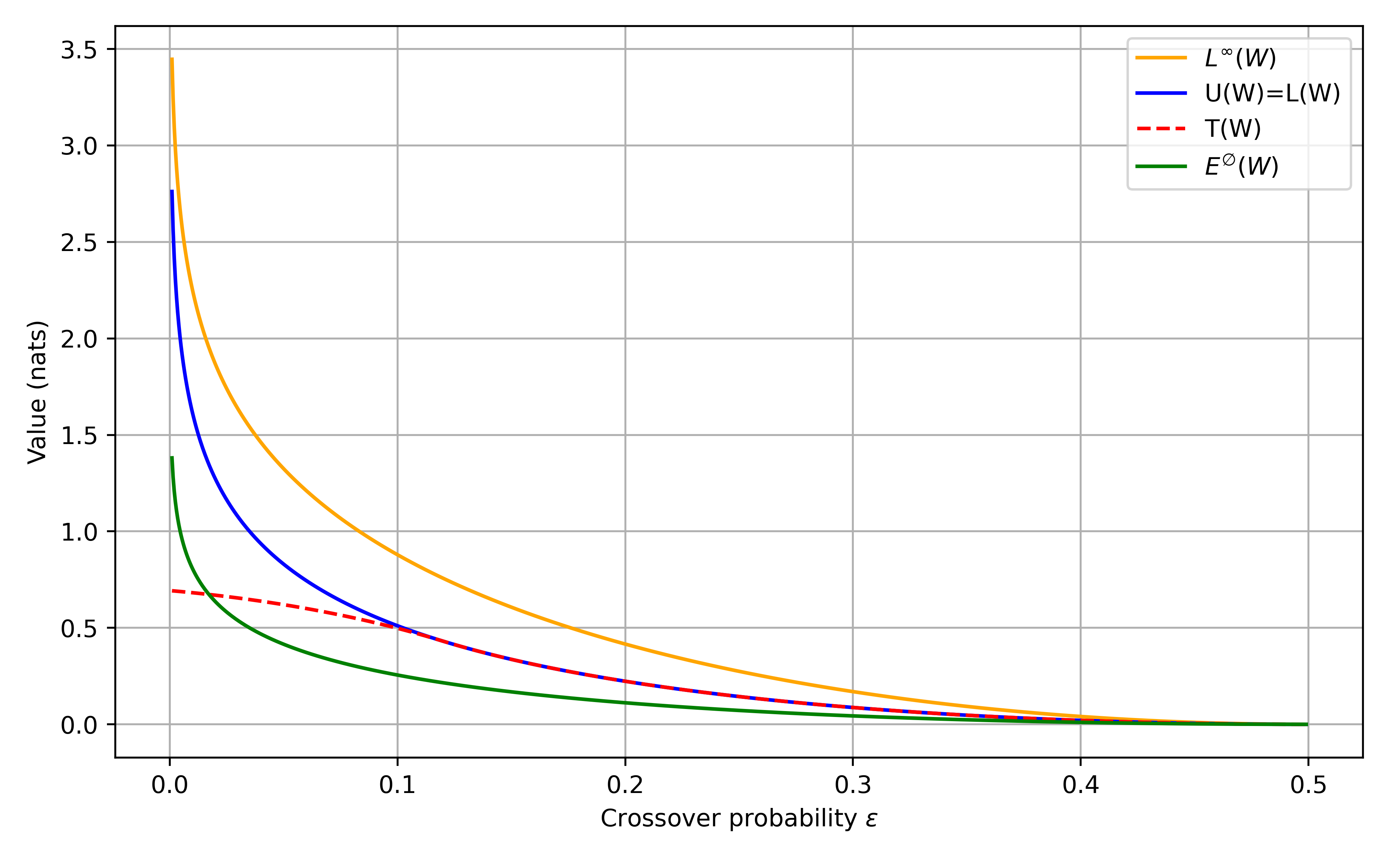}
        \caption{Regularized lautum information (yellow), umlaut information and lautum information (blue), tumula information (red), and zero-rate unassisted error exponent (green) of the binary symmetric channel $\pazocal{W}$ in terms of the crossover probability $\varepsilon$. Except for the tumula information, these quantities have been previously considered in~\cite{Filippo25, Lautum_08}. The bifurcation point between the umlaut and tumula information lies at $\varepsilon = \frac{1}{1+e^2}\approx 0.1192$. For larger crossover probabilities, the two quantities coincide. For smaller crossover probabilities, a phase transition happens and the optimizer $q^\ast$ is no longer unique. In fact, in the limit of vanishing error, it approaches a deterministic distribution and one finds $\displaystyle\lim_{\varepsilon \to 0} T (\pazocal W) = \log 2\approx 0.69$. }
        \label{fig:BSC}
\end{figure}

\begin{prop}[(Regularized tumula information of the identity channel)]\label{thm:tu_id_channel}
    Consider the classical identity channel $\pazocal{W} = I$. Then 
    $$T(I) = T^\infty(I) = \log |\mathcal{X}|.$$ 
\end{prop}
\begin{proof}
    See Appendix \ref{proof:id}.
\end{proof}
This then implies a general upper bound for the (regularized) channel tumula information.
\begin{cor}[(Upper bound on (regularized) channel tumula information)]\label{cor:channel_upperbound}
    For any classical channel $\pazocal{W}$
    \begin{equation}
        T(\pazocal{W}) \leq T^\infty(\pazocal{W}) \leq \log|\mathcal{X}| \,.
    \end{equation}
\end{cor}
\begin{proof}
     The corollary follows from Proposition~\ref{thm:tu_id_channel} and Theorem~\ref{thm:properties_petz} (monotonicity under local operations).
\end{proof}

At first glance, the finiteness of the (regularized) channel tumula information as expressed in Corollary~\ref{cor:channel_upperbound} is quite surprising, as the zero-rate unassisted error exponent and the umlaut information diverges for some channels. 
The umlaut information of a classical channel can be connected to the meta-converse bound in~\cite{PPV} in the zero-rate limit. 
This bound is achieved by means of non-signalling assisted codes~\cite{Matthews2012}, yielding an operational interpretation of the umlaut information. 
This bound on the number of messages that can be sent using a noisy (classical) channel is phrased and proved in terms of a hypothesis testing problem. 
Our results imply that it is impossible to connect the channel tumula information to a meta-converse analogous to the one in~\cite{PPV}.

Indeed, since the zero-rate unassisted error exponent $E^\emptyset(0^+,I)$ of the identity channel is infinite, Proposition \ref{thm:tu_id_channel} exhibits a strict gap between this exponent and the (regularized) channel tumula information, as $T(I)=T^\infty(I)=\log |\mathcal{X}|<+\infty =E^\emptyset(0^+,I)$. 
As a consequence, the tumula information of a channel cannot identify any error exponent in an assisted-communication setting. 

However, if one considers quantum channels instead of classical channels, the quantum channel tumula information can be infinite, as illustrated by the noiseless quantum channel. 
This raises the question whether the quantum channel tumula information admits other operational interpretations in quantum channel coding. 
We leave this question for future work.

\section{Acknowledgments} 
The authors thank Ludovico Lami for helpful discussions.
L.B. and L.S. acknowledge support from the National Centre of Competence in Research SwissMAP, the Quantum Center at ETH Zurich, the SNSF project No. 20QU-1\_225171, and the CHIST-ERA project MoDIC. 
F.G. acknowledges financial support from the European Union (ERC StG ETQO, Grant Agreement no.\ 101165230).

\appendix
\section{Proofs of properties}
\subsection{Proof of \eqref{eq:prli-limits}}\label{proof:limits}

\begin{proof}
Since $\displaystyle{\lim_{\alpha\to 1} D_\alpha(\rho\|\sigma)=D(\rho\|\sigma)}$, we immediately have
$L(A:B)_\rho = \lim_{\alpha\to 1}L_\alpha^{\uparrow\uparrow}(A:B)_\rho$ 
by the very definition of these quantities.

Let us first consider the limit where $\alpha\rightarrow 1$ from above. 
Since $\alpha\mapsto D_\alpha(\rho\|\sigma)$ is a monotonically increasing function, we can write
\bb
 \lim_{\alpha\to 1^+}L_\alpha^{\uparrow\downarrow}(A:B)_\rho
 &=\inf_{\substack{\alpha>1,\\ \tau_B\in \mathcal{S}(B)}} D_\alpha (\rho_A\otimes \tau_B\|\rho_{AB})
 =\inf_{\tau_B\in \mathcal{S}(B)}D (\rho_A\otimes \tau_B\|\rho_{AB})
 =U(A:B)_\rho,
\ee
and similarly
\bb
\lim_{\alpha\to 1^+}L_\alpha^{\downarrow\downarrow}(A:B)_\rho 
&= \inf_{\substack{\alpha>1,\\ \sigma_A\in \mathcal{S}(A), \\ \tau_B\in \mathcal{S}(B)}} D_\alpha (\sigma_A\otimes \tau_B\| \rho_{AB})
= \inf_{\substack{\sigma_A\in \mathcal{S}(A), \\ \tau_B\in \mathcal{S}(B)}} D (\sigma_A\otimes \tau_B\|\rho_{AB})=T(A:B)_\rho.
\ee

    Now, let us consider the limit where $\alpha\rightarrow 1$ from below. 
    Leveraging again the monotonicity of $\alpha\mapsto D_\alpha(\rho\|\sigma)$, by the Mosonyi--Hiai minimax theorem~\cite[Corollary A2]{MosonyiHiai}, we can rewrite
    \bb
 \lim_{\alpha\to 1^-}L_\alpha^{\uparrow\downarrow}(A:B)_\rho
 =\sup_{\alpha<1}\inf_{\tau_B\in \mathcal{S}(B)} D_\alpha (\rho_A\otimes \tau_B\| \rho_{AB})
 &=\inf_{\tau_B\in \mathcal{S}(B)}\sup_{\alpha<1} D_\alpha (\rho_A\otimes \tau_B\| \rho_{AB})\\
 &=\inf_{\tau_B\in \mathcal{S}(B)}D (\rho_A\otimes \tau_B\| \rho_{AB})
 =U(A:B)_\rho,
\ee
and similarly
\bb
\lim_{\alpha\to 1^-}L_\alpha^{\downarrow\downarrow}(A:B)_\rho 
&= \inf_{\substack{\sigma_A\in \mathcal{S}(A), \\ \tau_B\in \mathcal{S}(B)}} \sup_{\alpha<1} D_\alpha (\sigma_A\otimes \tau_B\| \rho_{AB})
= \inf_{\substack{\sigma_A\in \mathcal{S}(A), \\ \tau_B\in \mathcal{S}(B)}} D (\sigma_A\otimes \tau_B\|\rho_{AB})=T(A:B)_\rho.
\ee
\end{proof}

\subsection{Proof of Theorem~\ref{thm:properties_petz}}\label{proof:properties2}
\begin{proof}[Proof: Varia]
The following properties follow from corresponding properties of the doubly minimized PRMI~\cite{burri2025doubly} due to the correspondence between the doubly minimized PRMI and the doubly minimized LRMI, see~\eqref{eq:prli_prmi2}: 
Additivity for $\alpha\in (0,1/2)$ (additivity for $\alpha\in \{0,1/2\}$ follows from continuity), 
joint convexity, 
uniqueness of minimizer, 
asymptotic optimality of universal permutation invariant state, 
partial minimizers.
\end{proof}

\begin{proof}[Proof: Monotonicity under local operations]
Let $\alpha\in [0,2]$. 

\textbf{Case 1}: $L_\alpha^{\downarrow\downarrow}(A:B)_{\rho}<\infty$. 
Let
\begin{equation}
(\sigma_{A}^\star, \tau_{B}^\star)\in \argmin_{(\sigma_A,\tau_B)\in \mathcal{S}(A)\times \mathcal{S}(B)} D_\alpha(\sigma_{A}\otimes \tau_{B}\| \rho_{AB}).    
\end{equation}
Then, 
\begin{align}
    L_\alpha^{\downarrow\downarrow}(A':B')_{\pazocal{N} \otimes \pazocal{M} (\rho)} 
    &= \inf_{\substack{\sigma_{A'}\in \mathcal{S}(A'), \\ \tau_{B'}\in \mathcal{S}(B')}} D_\alpha(\sigma_{A'} \otimes \tau_{B'} \| \pazocal{N} \otimes \pazocal{M} (\rho_{AB})) \\
    &\leq D_\alpha(\pazocal{N}(\sigma_A^\star)\otimes\pazocal{M}(\tau_B^\star) \| \pazocal{N} \otimes \pazocal{M} (\rho_{AB})) \\
    &\leq D_\alpha(\sigma_A^\star \otimes \tau_B^\star \|\rho_{AB}) 
    = L_\alpha^{\downarrow\downarrow}(A:B)_\rho
\end{align}

\textbf{Case 2}: $L_\alpha^{\downarrow\downarrow}(A:B)_{\rho}=\infty$. Then the claim is trivially true.
\end{proof}

\begin{proof}[Proof of continuity in $\alpha$]
The continuity of $L_\alpha^{\downarrow\downarrow}(A:B)_\rho$ on $\alpha\in [0,1)$ and on $\alpha\in [1,\infty)$ follows from the continuity in $\alpha$ of the Petz divergence. 
It remains to prove left-continuity at $\alpha=1$. 
By \eqref{eq:l1-omegaomega}, we have for any $n\in \mathbb{N}_{>0}$
\begin{align}\label{eq:proof-prmi2-i1}
    &\left(1-\frac{1}{\sqrt{n}}\right)\frac{1}{\sqrt{n}}D_{\frac{1}{\sqrt{n}}}(\rho_{AB}^{\otimes n}\| \omega_{A^n}^n\otimes \omega_{B^n}^n)-\left(1-\frac{1}{\sqrt{n}}\right)\frac{\log (g_{n,d_A}g_{n,d_B})}{\sqrt{n}}
    \\
    &\leq L_{1-\frac{1}{\sqrt{n}}}^{\downarrow\downarrow}(A:B)_\rho 
    \leq \lim_{\alpha\rightarrow 1^-}L_\alpha^{\downarrow\downarrow}(A:B)_\rho
    \leq L_1^{\downarrow\downarrow}(A:B)_\rho,
    \label{eq:proof-prmi2-i2}
\end{align}
where the last two inequalities follows from the monotonicity in $\alpha$.
The second term in~\eqref{eq:proof-prmi2-i1} vanishes in the limit  $n\rightarrow\infty$~\cite[Proposition~1]{burri2025doubly}. Therefore,
\begin{align}\label{eq:prmi2-proof-left-c}
L_1^{\downarrow\downarrow}(A:B)_\rho
=\lim_{n\rightarrow\infty}\frac{1}{n}D_1(\rho_{AB}^{\otimes n}\| \omega_{A^n}^n\otimes \omega_{B^n}^n)
\leq \lim_{\alpha\rightarrow 1^-}L_\alpha^{\downarrow\downarrow}(A:B)_\rho
\leq L_1^{\downarrow\downarrow}(A:B)_\rho,
\end{align}
where the first equality in~\eqref{eq:prmi2-proof-left-c} follows from~\eqref{eq:proof-prmi2-i2}. 
Thus, $\lim_{\alpha\rightarrow 1^-}L_\alpha^{\downarrow\downarrow}(A:B)_\rho=L_1^{\downarrow\downarrow}(A:B)_\rho$.
\end{proof}
\begin{proof}[Proof of convexity in $\alpha$]
    Let us denote the given function by $f$, and let us define $g:[0,1]\rightarrow\mathbb{R},\alpha\mapsto (\alpha-1)I_\alpha^{\downarrow\downarrow}(A:B)_\rho$. 
    Then, $f(\alpha)=g(1-\alpha)$ for all $\alpha\in [0,1)$. 
    Since $g$ is convex~\cite{burri2025doubly} and $\alpha\mapsto 1-\alpha$ is an affine transformation, it follows that $f$ is convex.
\end{proof}

\begin{proof}[Proof of continuous differentiability]
The following proof technique is adapted from~\cite{burri2025doubly}. 
Let us define the following two functions.
\begin{align}
f:\quad (0,1/2)\rightarrow \mathbb{R},\quad \alpha
&\mapsto L_\alpha^{\downarrow\downarrow}(A:B)_\rho
\\
g:\quad (0,1/2)\rightarrow \mathbb{R},\quad \alpha
&\mapsto (\alpha -1)L_\alpha^{\downarrow\downarrow}(A:B)_\rho
\end{align}
$f$ is continuous, and 
$g$ is convex and continuous. 
The convexity of $g$ implies that the left and right derivative of $g$ exist at all points within its domain. 

For any $\alpha\in (0,\frac{1}{2})$ and 
any fixed $(\sigma_A,\tau_B)\in \argmin_{(\sigma_A',\tau_B')\in \mathcal{S}(A)\times \mathcal{S}(B)}D_\alpha (\sigma_A'\otimes \tau_B'\| \rho_{AB})$
\begin{subequations}\label{eq:prmi2-diff0}
\begin{align}
\frac{\partial}{\partial \alpha^+}f(\alpha)
&=\lim_{\varepsilon\rightarrow 0^+} \frac{1}{\varepsilon} (L_{\alpha+\varepsilon}^{\downarrow\downarrow}(A:B)_\rho -L_{\alpha}^{\downarrow\downarrow}(A:B)_\rho )
\\
&\leq \lim_{\varepsilon\rightarrow 0^+} \frac{1}{\varepsilon} (D_{\alpha+\varepsilon}(\sigma_A\otimes \tau_B\| \rho_{AB}) -D_{\alpha}(\sigma_A\otimes \tau_B\| \rho_{AB}) )
\\
&=\frac{\partial}{\partial \alpha}D_{\alpha}(\sigma_A\otimes \tau_B\| \rho_{AB})
\label{eq:prmi2-diff1}\\
&=\lim_{\varepsilon\rightarrow 0^-} \frac{1}{\varepsilon} (D_{\alpha+\varepsilon}(\sigma_A\otimes \tau_B\| \rho_{AB}) -D_{\alpha}(\sigma_A\otimes \tau_B\| \rho_{AB}) )
\label{eq:prmi2-diff2}\\
&\leq \lim_{\varepsilon\rightarrow 0^-} \frac{1}{\varepsilon} (L_{\alpha+\varepsilon}^{\downarrow\downarrow}(A:B)_\rho -L_{\alpha}^{\downarrow\downarrow}(A:B)_\rho )
=\frac{\partial}{\partial \alpha^-}f(\alpha).
\end{align}
\end{subequations}
The equalities in~\eqref{eq:prmi2-diff1} and~\eqref{eq:prmi2-diff2} hold due to the differentiability in $\alpha$ of the Petz divergence. 

For any $\alpha\in (0,\frac{1}{2})$, let 
\begin{equation}
(\sigma_A^{(\alpha)},\tau_B^{(\alpha)})\in \argmin_{(\sigma_A,\tau_B)\in \mathcal{S}(A)\times \mathcal{S}(B)}D_\alpha (\sigma_A\otimes \tau_B\| \rho_{AB})
\end{equation}
denote the unique minimizer. 
Let us define the function 
\begin{align}
(0,1/2)\rightarrow \mathbb{R}, \quad 
\alpha\mapsto 
h(\alpha)\coloneqq 
f(\alpha)+(\alpha -1) 
\frac{\partial}{\partial \alpha}D_\alpha (\sigma_A^{(\alpha)}\otimes \tau_B^{(\alpha)}\| \rho_{AB}),
\end{align}
where $\sigma_A^{(\alpha)}$ and $\tau_B^{(\alpha)}$ are kept fixed.  

The map $\alpha\mapsto (\sigma_A^{(\alpha)},\tau_B^{(\alpha)})$ is continuous on $\alpha\in (0,\frac{1}{2})$ 
due to the uniqueness of $(\sigma_A^{(\alpha)},\tau_B^{(\alpha)})$. 
By the continuous differentiability of the Petz divergence, it follows that $h$ is continuous. 
For any $\alpha\in (0,\frac{1}{2})$,
\begin{align}
\frac{\partial}{\partial \alpha^+}f(\alpha)
&\leq \frac{\partial}{\partial \alpha}D_\alpha (\sigma_A^{(\alpha)}\otimes \tau_B^{(\alpha)}\| \rho_{AB})
\leq \frac{\partial}{\partial \alpha^-}f(\alpha),
\label{eq:diff-h0}\\ 
\frac{\partial}{\partial \alpha^-}g(\alpha)
&\leq h(\alpha)
\leq \frac{\partial}{\partial \alpha^+}g(\alpha).
\label{eq:diff-g0}
\end{align}
\eqref{eq:diff-h0} follows from~\eqref{eq:prmi2-diff0}, and it is understood that $\sigma_A^{(\alpha)}$ and $ \tau_B^{(\alpha)}$ are kept fixed in~\eqref{eq:diff-h0}. 
\eqref{eq:diff-g0} follows from~\eqref{eq:diff-h0}. 
Therefore, for any $\alpha\in (0,\frac{1}{2})$, 
\begin{subequations}\label{eq:prmi2-g0}
\begin{align}
h(\alpha)
=\lim_{\varepsilon\rightarrow 0^+}h(\alpha-\varepsilon)
&\leq \lim_{\varepsilon\rightarrow 0^+}\frac{\partial}{\partial \alpha^+}g(\alpha-\varepsilon)
\label{eq:prmi2-g1}\\
&\leq \frac{\partial}{\partial \alpha^+}g(\alpha)
\leq \lim_{\varepsilon\rightarrow 0^+}\frac{\partial}{\partial \alpha^-}g(\alpha+\varepsilon)
\leq\lim_{\varepsilon\rightarrow 0^+}h(\alpha+\varepsilon)
=h(\alpha).
\label{eq:prmi2-g2}
\end{align}
\end{subequations}
The first two inequalities in~\eqref{eq:prmi2-g2} follow from the convexity of $g$. 
It follows that all inequalities in~\eqref{eq:prmi2-g0} must be saturated, so 
$\frac{\partial}{\partial \alpha^+}g(\alpha)=h(\alpha)$ for all $\alpha\in (0,\frac{1}{2})$. 
Since $h$ is continuous, also $\frac{\partial}{\partial \alpha^+}g(\alpha)$ is continuous on $\alpha\in (0,\frac{1}{2})$. 
Since $g$ is convex, the continuity of the right derivative of $g$ implies that $g$ is differentiable and $g'(\alpha)=\frac{\partial}{\partial \alpha^+}g(\alpha)=h(\alpha)$ for all $\alpha\in (0,\frac{1}{2})$. 
Since $h$ is continuous, this proves that 
$g$ is \emph{continuously} differentiable on $\alpha\in (0,\frac{1}{2})$. 
By the product rule, this implies that also 
$f$ is continuously differentiable on $\alpha\in (0,\frac{1}{2})$. 

The additional claim in~\eqref{eq:diff-prmi2} follows from the relation between the doubly minimized PRLI and the doubly minimized PRMI.
\end{proof}

\subsection{Proof of Theorem~\ref{thm:properties}}\label{proof:properties}
\begin{proof}
    [Proof: Partial minimizers a)]
    To find the partial minimizers of
    \begin{equation}
        \min_{\sigma_A,\tau_B} D(\sigma_A \otimes \tau_B \|\rho_{AB}) =  \min_{\sigma_A,\tau_B} \left( \Tr[\sigma_A\log\sigma_A] + \Tr[\tau_B \log \tau_B] - \Tr[\sigma_A\otimes\tau_B \log \rho_{AB}] \right) \, ,
    \end{equation}
    we formulate a Lagrangian
    \begin{equation}
        L(\sigma_A,\tau_B,\lambda_A,\lambda_B) = \Tr[\sigma_A\log\sigma_A] + \Tr[\tau_B \log \tau_B] - \Tr[\sigma_A\otimes\tau_B \log \rho_{AB}] + \lambda_A (\Tr[\sigma_A]-1) + \lambda_B (\Tr[\tau_B]-1) 
    \end{equation}
    Taking Fr\'echet derivatives for $\sigma$ and $\tau$, while fixing the other, and demanding that they vanish leads to
    \begin{equation}
        \log \sigma_A + I - \Tr_B[\tau_B \log \rho_{AB}] + \lambda_A I = 0
    \end{equation}
    and 
    \begin{equation}
        \log \tau_B + I - \Tr_A[\sigma_A \log \rho_{AB}] + \lambda_B I = 0 \, .
    \end{equation}
    Which gives the self-consistency equations 
    \begin{equation}
    \sigma_A^\star =\frac{\exp[\Tr_B[\tau_B\log\rho_{AB}]]}{\Tr[\exp[\Tr_B[\tau_B\log\rho_{AB}]]]}
    \end{equation}
    and 
    \begin{equation}
    \tau_B^\star =\frac{\exp[\Tr_A[\sigma_A\log\rho_{AB}]]}{\Tr[\exp[\Tr_A[\sigma_A\log\rho_{AB}]]]} \, .
    \end{equation}
\end{proof}

\begin{proof}
[Proof: Additivity b)]
Since we minimize over set of different size, we have
\begin{align}
    T(A_1A_2:B_1B_2)_{\rho_{A_1B_1}\otimes \rho'_{A_2B_2}} 
    &\leq T(A_1 : B_1)_{\rho_{A_1B_1}} + T(A_2:B_2)_{\rho'_{A_2B_2}}
\end{align}
because $\mathcal{S}(A_1A_2)\supseteq \{\sigma_{A_1}\otimes \sigma_{A_2}'\}_{(\sigma_{A_1},\sigma_{A_2}')\in \mathcal{S}(A_1)\times \mathcal{S}(A_2)}$ and 
$\mathcal{S}(B_1B_2)\supseteq \{\tau_{B_1}\otimes \tau_{B_2}'\}_{(\tau_{B_1},\tau_{B_2}')\in \mathcal{S}(A_1)\times \mathcal{S}(B_2)}$. 
It remains to prove the opposite inequality. \\

\textbf{Case 1}: $T(A_1A_2:B_1B_2)_{\rho\otimes \rho'}<\infty$. 
Choose minimizers 
\begin{equation}
(\sigma_{A_1A_2}^\star, \tau_{B_1B_2}^\star)\in \argmin_{(\sigma_{A_1A_2}, \tau_{B_1B_2})\in \mathcal{S}(A_1A_2)\times \mathcal{S}(B_1B_2)} D(\sigma_{A_1A_2}\otimes \tau_{B_1B_2}\| \rho_{A_1B_1} \otimes \rho_{A_2B_2}').    
\end{equation}
Then, 
\begin{align}
    &T(A_1A_2:B_1B_2)_{\rho\otimes \rho'} 
    \\
    &= D(\sigma_{A_1A_2}^\star\otimes \tau_{B_1B_2}^\star\|\rho_{A_1B_1} \otimes \rho_{A_2B_2}') \\ 
    &=-H(A_1A_2)_{\rho\otimes \rho'} - H(B_1B_2)_{\rho\otimes \rho'} - \Tr[\sigma_{A_1}^\star\otimes\tau_{B_1}^\star \log \rho_{A_1B_1}] - \Tr[\sigma_{A_2}^\star\otimes\tau_{B_2}^\star \log \rho_{A_2B_2}'] \\
    &\geq -H(A_1)_{\rho}-H(A_2)_{\rho'}-H(B_1)_{\rho}-H(B_2)_{\rho'} - \Tr[\sigma_{A_1}^\star\otimes\tau_{B_1}^\star \log \rho_{A_1B_1}] - \Tr[\sigma_{A_2}^\star\otimes\tau_{B_2}^\star \log \rho_{A_2B_2}'] \\
    &\geq T(A_1 : B_1)_{\rho} + T(A_2:B_2)_{\rho'} \, ,
\end{align}
where the first inequality follows from the subadditivity of von Neumann entropy, and the second inequality from the minimization in the definition of the tumula information.

\textbf{Case 2}: $T(A_1A_2:B_1B_2)_{\rho\otimes \rho'}=\infty$. 
Then, the opposite inequality is trivially true.
\end{proof}

\begin{proof}[Proof: Monotonicity under local operations c)]
\textbf{Case 1}: $T(A:B)_{\rho}<\infty$. 
Let
\begin{equation}
(\sigma_{A}^\star, \tau_{B}^\star)\in \argmin_{(\sigma_A,\tau_B)\in \mathcal{S}(A)\times \mathcal{S}(B)} D(\sigma_{A}\otimes \tau_{B}\| \rho_{AB}).    
\end{equation}
Then, 
\begin{align}
    T(A':B')_{\pazocal{N} \otimes \pazocal{M} (\rho)} 
    &= \inf_{\substack{\sigma_{A'}\in \mathcal{S}(A'), \\ \tau_{B'}\in \mathcal{S}(B')}} D(\sigma_{A'} \otimes \tau_{B'} \| \pazocal{N} \otimes \pazocal{M} (\rho_{AB})) \\
    &\leq D(\pazocal{N}(\sigma_A^\star)\otimes\pazocal{M}(\tau_B^\star) \| \pazocal{N} \otimes \pazocal{M} (\rho_{AB})) \\
    &\leq D(\sigma_A^\star \otimes \tau_B^\star \|\rho_{AB}) 
    = T(A:B)_\rho
\end{align}

\textbf{Case 2}: $T(A:B)_{\rho}=\infty$. Then the claim is trivially true.
\end{proof}

\begin{proof}[Proof: Universal permutation invariant state d)] 
Let $\varepsilon\in (0,\frac{1}{2}),n\in \mathbb{N}_{>0}$. On the one hand, 
\begin{align}
    L_{1-\varepsilon}^{\downarrow\downarrow}(A:B)_\rho 
    &=\frac{1-\varepsilon}{\varepsilon}I_\varepsilon^{\downarrow\downarrow}(A:B)_\rho
    \leq \frac{1-\varepsilon}{\varepsilon}\frac{1}{n}D_{\varepsilon}(\rho_{AB}^{\otimes n}\| \omega_{A^n}^n\otimes \omega_{B^n}^n)+\frac{\log (g_{n,d_A}g_{n,d_B})}{n},
\end{align}
where the inequality follows from~\cite{burri2025doubly}. 
On the other hand,
\begin{align}
    L_{1-\varepsilon}^{\downarrow\downarrow}(A:B)_\rho 
    &=\frac{1-\varepsilon}{\varepsilon}I_\varepsilon^{\downarrow\downarrow}(A:B)_\rho
    \geq \frac{1-\varepsilon}{\varepsilon}\frac{1}{n}D_{\varepsilon}(\rho_{AB}^{\otimes n}\| \omega_{A^n}^n\otimes \omega_{B^n}^n)-\frac{1-\varepsilon}{\varepsilon}\frac{\log (g_{n,d_A}g_{n,d_B})}{n}
\end{align}
where the inequality follows from~\cite{burri2025doubly}. 
Evaluating these inequalities for $\varepsilon=\frac{1}{\sqrt{n}}$ and assuming $n>4$ implies that
\begin{align}
    L_{1-\frac{1}{\sqrt{n}}}^{\downarrow\downarrow}(A:B)_\rho 
    &\leq \left(1-\frac{1}{\sqrt{n}}\right)\frac{1}{\sqrt{n}}D_{\frac{1}{\sqrt{n}}}(\rho_{AB}^{\otimes n}\| \omega_{A^n}^n\otimes \omega_{B^n}^n)+\frac{\log (g_{n,d_A}g_{n,d_B})}{n},
    \\
    L_{1-\frac{1}{\sqrt{n}}}^{\downarrow\downarrow}(A:B)_\rho 
    &\geq \left(1-\frac{1}{\sqrt{n}}\right)\frac{1}{\sqrt{n}}D_{\frac{1}{\sqrt{n}}}(\rho_{AB}^{\otimes n}\| \omega_{A^n}^n\otimes \omega_{B^n}^n)-\left(1-\frac{1}{\sqrt{n}}\right)\frac{\log (g_{n,d_A}g_{n,d_B})}{\sqrt{n}}.
    \label{eq:l1-omegaomega}
\end{align}
In the previous two lines, the terms on the right-hand side vanish as $n\rightarrow\infty$, see~\cite[Proposition~1]{burri2025doubly}. 
Thus, 
\begin{align}
    T(A:B)_\rho 
    &= \lim_{n\rightarrow\infty} L_{1-\frac{1}{\sqrt{n}}}^{\downarrow\downarrow}(A:B)_\rho 
    = \lim_{n\rightarrow\infty} \frac{1}{\sqrt{n}} D_{\frac{1}{\sqrt{n}}}(\rho_{AB}^{\otimes n}\| \omega_{A^n}^n\otimes \omega_{B^n}^n).
\end{align}
\end{proof}

\begin{proof}[Classical states]
    The assertion follows from Sibson's identity~\eqref{eq:sibson}.
\end{proof}

\subsection{Proof of Proposition~\ref{prop:bound}}\label{proof:bound}
\begin{proof}
Let $P_{XY}$ be the joint probability distribution of $X$ and $Y$. Then,
\bb
    T(X:Y)&=\min_{Q_X}\min_{Q_Y} D(Q_XQ_Y\|P_{XY})\\
    &=\min_{Q_X}\min_{Q_Y}\Big( -H(Q_X)-H(Q_Y)-\sum_{x\in \mathcal{X},y\in \mathcal{Y}}Q_X(x)Q_Y(y)\log P_{XY}(x,y)\Big)\\
    &\leqt{(i)} \min_{Q_Y}\min_{Q_X}\Big(-H(Q_Y)-\sum_{x\in \mathcal{X},y\in \mathcal{Y}}Q_X(x)Q_Y(y)\log P_{XY}(x\in \mathcal{X},y\in \mathcal{Y})\Big)\\
    &=\min_{Q_X}\Big(-H(Q_Y)-\max_{Q_X}\sum_{x\in \mathcal{X},y\in \mathcal{Y}}Q_X(x)Q_Y(y)\log P_{XY}(x,y)\Big)\\
    &=\min_{Q_X}\Big(-H(Q_Y)-\max_{x\in \mathcal{Y}}\sum_{y\in \mathcal{Y}}Q_Y(y)\log P_{XY}(x,y)\Big)\\
    &=\min_{x\in \mathcal{Y}}\min_{Q_X}\Big(-H(Q_Y)-\sum_{y\in \mathcal{Y}}Q_Y(y)\log P_{XY}(x,y)\Big)\\
    &= \min_{x\in \mathcal{X}}\Big(-\log\sum_{x,y}P_{XY}(x,y)\Big)
    = -\log \max_{x\in \mathcal{X}}P_{X}(x)
    \leqt{(ii)} \log|\mathcal{X}|,
\ee
where (i) follows from Gibb's variational principle, and in (ii) we have noticed that there always exists a symbol $x\in\mathcal{X}$ such that $P_X(x)\geq 1/|\mathcal{X}|$. 
Since $T$ is symmetric, we also immediately get $T(X:Y)\leq \log |\mathcal{Y}|$. 
Now, without loss of generality, let us suppose $\mathcal{X}\subseteq\mathcal{Y}$.
\bb
    \bar P_{XY}(x,y)\coloneqq \frac{1}{|\mathcal{X}|}\cdot\begin{cases}
        1 & x= y,\\
        0 & x\neq  y.
    \end{cases}
\ee
Then, we have that $D(Q_XQ_Y\|\bar P_{XY})\iff Q_XQ_Y\ll \bar P_{XY}$: this means that the minimisation will select $Q_X$ and $Q_Y$ such that $\text{supp}(Q_X) \times \text{supp}(Q_Y) \subseteq \text{supp}(\bar P_{XY})$. This is only possible if both $Q_x$ and $Q_Y$ are Dirac deltas; therefore, we get $T(\bar X:\bar Y)=\log |\mathcal{X}|$. Again by symmetry, we conclude that the bound is tight also when $|\mathcal{Y}|\leq|\mathcal{X}|$.
\end{proof}

\section{Proofs for hypothesis testing}
\subsection{Proof of Theorem~\ref{thm:singly}}\label{proof:singly}

Based on the three possible choices for the null hypothesis, 
let us define the following functions of $\mu\in [0,\infty)$ 
for any $\rho_{AB}\in \mathcal{S}(AB),n\in \mathbb{N}_{>0}$.
\begin{align}
\hat{\beta}_{n}^{\mathrm{iid}}(\mu )
&\coloneqq \min_{\substack{T^n_{A^nB^n}\in \mathcal{L} (A^nB^n): \\ 0\leq T^n_{A^nB^n}\leq 1}}
\{\max_{\tau_{B} \in \mathcal{S}(B)}
\Tr[\rho_A^{\otimes n}\otimes \tau_{B}^{\otimes n} T^n_{A^nB^n}]:
\Tr[\rho_{AB}^{\otimes n}(1-T^n_{A^nB^n})]
\leq \mu\}
\label{eq:def-beta-singly1}
\\
\hat{\beta}_{n}^{\sym}(\mu )
&\coloneqq \min_{\substack{T^n_{A^nB^n}\in \mathcal{L} (A^nB^n): \\ 0\leq T^n_{A^nB^n}\leq 1}}
\{\max_{\tau_{B^n} \in \mathcal{S}_{\sym} (B^n)}
\Tr[\rho_A^{\otimes n}\otimes \tau_{B^n} T^n_{A^nB^n}]:
\Tr[\rho_{AB}^{\otimes n}(1-T^n_{A^nB^n})]
\leq \mu\}
\label{eq:def-beta-singly2}
\\
\hat{\beta}_{n}^{\mathrm{ind}}(\mu )
&\coloneqq \min_{\substack{T^n_{A^nB^n}\in \mathcal{L} (A^nB^n): \\ 0\leq T^n_{A^nB^n}\leq 1}}
\{\max_{\tau_{B^n} \in \mathcal{S}(B^n)}
\Tr[\rho_A^{\otimes n}\otimes \tau_{B^n} T^n_{A^nB^n}]:
\Tr[\rho_{AB}^{\otimes n}(1-T^n_{A^nB^n})]
\leq \mu\}
\label{eq:def-beta-singly3}
\end{align} 

The following lemma describes their natural ordering. 
\begin{lemma}[(Minimum type-I errors)]\label{lem:order}
Let $\rho_{AB}\in \mathcal{S}(AB),n\in \mathbb{N}_{>0}$. Then, for all $\mu\in [0,\infty)$
\begin{align}
    0&\leq \hat{\beta}_{n}^{\mathrm{iid}} (\mu)
\leq \hat{\beta}_{n}^{\sym}(\mu)
= \hat{\beta}_{n}^{\mathrm{ind}}(\mu)
\leq \max(0,1-\mu) \leq 1\,.
\end{align}
\end{lemma}
The proof of this lemma is straightforward. We use the same proof technique that was used in an analogous proof in~\cite{burri2025doubly}.
\begin{proof}
The first inequality holds because any feasible $T^n_{A^nB^n}$ is such that $T^n_{A^nB^n}\leq 1$, which implies that $(1-T^n_{A^nB^n})\geq 0$. 

The second inequality follows from $\{\tau_B^{\otimes n}\}_{\tau_B\in\mathcal{S}(B)} \subseteq \mathcal{S}_{\sym} (B^n)$.

The inequality $\hat{\beta}_{n}^{\sym}(\mu)
\leq \hat{\beta}_{n}^{\mathrm{ind}}(\mu)$ follows from $\mathcal{S}_{\sym} (B^n)\subseteq \mathcal{S}(B^n)$.

The second last inequality follows from choosing the test $T^n_{A^nB^n}\coloneqq \mu 1$ if $\mu\in [0,1]$, and $T^n_{A^nB^n}\coloneqq 1$ if $\mu\in (1,\infty)$. 

The last inequality is trivial.

It remains to prove that $\hat{\beta}_{n}^{\sym}(\mu)
\geq \hat{\beta}_{n}^{\mathrm{ind}}(\mu)$.

Let $\mu\in [0,\infty)$. 
Let $T^n_{A^nB^n}\in \mathcal{L}(A^nB^n)$ be in the feasible set of the optimization problem that defines $\hat{\beta}_{n}^{\sym}(\mu)$. 
Let $\hat{T}^n_{A^nB^n}\coloneqq \frac{1}{|S_n|}\sum_{\pi \in S_n}U(\pi)_{A^n}\otimes U(\pi)_{B^n} T^n_{A^nB^n}U(\pi)_{A^n}^\dagger\otimes U(\pi)_{B^n}^\dagger$. 
Then, for any $\tau_{B^n} \in \mathcal{S}_{\sym}(B^{\otimes n})$
\begin{align}
\Tr[\rho_A^{\otimes n}\otimes \tau_{B^n}T^n_{A^nB^n}]
&=\frac{1}{|S_n|}\sum_{\pi \in S_n}\Tr[(U(\pi)_{A^n}^\dagger \rho_A^{\otimes n}U(\pi)_{A^n})\otimes ( U(\pi)_{B^n}^\dagger\tau_{B^n}U(\pi)_{B^n})T^n_{A^nB^n}]
\\
&=\frac{1}{|S_n|}\sum_{\pi \in S_n}\Tr[\rho_A^{\otimes n}\otimes\tau_{B^n} U(\pi)_{A^n}\otimes U(\pi)_{B^n} T^n_{A^nB^n}U(\pi)_{A^n}^\dagger\otimes U(\pi)_{B^n}^\dagger]
\\
&=\Tr[\rho_A^{\otimes n}\otimes \tau_{B^n}\hat{T}^n_{A^nB^n}].
\end{align}
Hence, 
\begin{align}
\max_{\tau_{B^n} \in \mathcal{S}_{\sym}(B^{\otimes n})}
\Tr[\rho_A^{\otimes n}\otimes \tau_{B^n}T^n_{A^nB^n}] 
=\max_{\tau_{B^n} \in \mathcal{S}_{\sym}(B^{\otimes n})}
\Tr[\rho_A^{\otimes n}\otimes \tau_{B^n}\hat{T}^n_{A^nB^n}] .
\end{align}
Since $\rho_{AB}^{\otimes n}\in \mathcal{S}_{\sym}((AB)^{\otimes n})$, we have 
$\Tr[\rho_{AB}^{\otimes n}(1-T^n_{A^nB^n})]=\Tr[\rho_{AB}^{\otimes n}(1-\hat{T}^n_{A^nB^n})]$.
Since $\hat{T}^n_{A^nB^n}$ is permutation invariant, 
it follows that
\begin{align}\label{eq:a-t-sym}
\hat{\beta}_{n}^{\sym}(\mu)
&=\min_{\substack{T^n_{A^nB^n}\in \mathcal{L}_{\sym}((AB)^{\otimes n}) : \\ 0\leq T^n_{A^nB^n}\leq 1}}
\{\max_{ \tau_{B^n} \in \mathcal{S}_{\sym}(B^{\otimes n})}
\Tr[\rho_A^{\otimes n}\otimes \tau_{B^n}T^n_{A^nB^n}]:
\Tr[\rho_{AB}^{\otimes n}(1-T^n_{A^nB^n})] \leq \mu \}.
\end{align}
Let $\widetilde{T}^n_{A^nB^n}\in \mathcal{L}_{\sym}((AB)^{\otimes n})$ be positive semidefinite.
Then, for all 
$\tau_{B^n} \in \mathcal{S}(B^n)$
\begin{align}
\Tr[\rho_A^{\otimes n}\otimes \tau_{B^n} \widetilde{T}^n_{A^nB^n}]
&=\frac{1}{|S_n|}\sum_{\pi\in S_n}
\Tr[\rho_A^{\otimes n}\otimes \tau_{B^n} (U(\pi)_{A^n}\otimes U(\pi)_{B^n} \widetilde{T}^n_{A^nB^n}  U(\pi)_{A^n}^\dagger\otimes U(\pi)_{B^n}^\dagger)]
\\
&=\frac{1}{|S_n|}\sum_{\pi\in S_n}
\Tr[\underbrace{(U(\pi)_{A^n}^\dagger\rho_A^{\otimes n} U(\pi)_{A^n})}_{=\rho_A^{\otimes n}} \otimes (U(\pi)_{B^n}^\dagger\tau_{B^n} U(\pi)_{B^n}) \widetilde{T}^n_{A^nB^n}]
\\
&=
\Tr[\rho_A^{\otimes n}\otimes \underbrace{\frac{1}{|S_n|}\sum_{\pi\in S_n}(U(\pi)_{B^n}^\dagger\tau_{B^n} U(\pi)_{B^n})}_{\in \mathcal{S}_{\sym}(B^{\otimes n})} \widetilde{T}^n_{A^nB^n}].
\end{align} 
Hence,
\begin{equation}\label{eq:max-symt}
\max_{\tau_{B^n} \in \mathcal{S}(B^n)}
\Tr[\rho_A^{\otimes n}\otimes \tau_{B^n} \widetilde{T}^n_{A^nB^n}]
=\max_{ \tau_{B^n} \in \mathcal{S}_{\sym}(B^{\otimes n})}
\Tr[\rho_A^{\otimes n}\otimes \tau_{B^n} \widetilde{T}^n_{A^nB^n}].
\end{equation}
We conclude that 
\begin{align}
\hat{\beta}_{n}^{\mathrm{ind}}(\mu)
&\leq \min_{\substack{T^n_{A^nB^n}\in \mathcal{L}_{\sym}((AB)^{\otimes n}): \\ 0\leq T^n_{A^nB^n}\leq 1}}
\{\max_{ \tau_{B^n} \in \mathcal{S}(B^n)}
\Tr[\rho_A^{\otimes n}\otimes \tau_{B^n}T^n_{A^nB^n}]:
\Tr[\rho_{AB}^{\otimes n}(1-T^n_{A^nB^n})] \leq \mu \}
\\
&= \min_{\substack{T^n_{A^nB^n}\in \mathcal{L}_{\sym}((AB)^{\otimes n}): \\ 0\leq T^n_{A^nB^n}\leq 1}}
\{\max_{\tau_{B^n} \in \mathcal{S}_{\sym}(B^{\otimes n})}
\Tr[\rho_A^{\otimes n}\otimes \tau_{B^n}T^n_{A^nB^n}]:
\Tr[\rho_{AB}^{\otimes n}(1-T^n_{A^nB^n})] \leq \mu \}
\label{eq:alpha-proof2}\\
&= \hat{\beta}_{n}^{\sym}(\mu).
\label{eq:alpha-proof3}
\end{align}
\eqref{eq:alpha-proof2} follows from~\eqref{eq:max-symt}.
\eqref{eq:alpha-proof3} follows from~\eqref{eq:a-t-sym}.
\end{proof}

The proof of Theorem~\ref{thm:singly} is divided into two parts: a proof of achievability and a proof of optimality. 
Their combination with the Lemma~\ref{lem:order} implies the claim.

\subsubsection{Proof of achievability}\label{sec:singly_achievability}
Below, we prove that for any $R\in (0,\infty)$
\begin{align}\label{eq:singly-ach}
\liminf_{n\rightarrow\infty}-\frac{1}{n}\log \hat{\beta}_{n}^{\sym}(e^{-nR})
\geq \sup_{s\in (0,1)}\frac{1-s}{s}(L_s^{\uparrow\downarrow}(A:B)_\rho - R).
\end{align}
\begin{proof}
Let $R\in (0,\infty)$ and $s\in (0,1)$ be arbitrary but fixed. For all $n\in \mathbb{N}_{>0}$, we define
\begin{align}
\lambda_n\coloneqq -\frac{1}{s} 
\left(\ nR-sD_{1-s}(\rho_{AB}^{\otimes n}\| \rho_A^{\otimes n} \otimes \omega_{B^n}^n )\right)
\label{eq:singly-lambda}
\end{align}
and the test
$T^n_{A^nB^n}\coloneqq \{\rho_{AB}^{\otimes n}\leq e^{\lambda_n}\rho_A^{\otimes n}\otimes \omega_{B^n}^n\}$.
For this test holds
\begin{subequations}\label{eq:singly-feasible}
\begin{align}
\Tr[\rho_{AB}^{\otimes n} T^n_{A^nB^n}]
&= \Tr[\rho_{AB}^{\otimes n}\{\rho_{AB}^{\otimes n}\leq e^{\lambda_n}\rho_A^{\otimes n}\otimes\omega_{B^n}^n \} ] 
\label{eq:ach-b0}\\
&\leq \Tr[(\rho_{AB}^{\otimes n})^{1-s} (e^{\lambda_n} \rho_A^{\otimes n}\otimes\omega_{B^n}^n )^{s}]
\label{eq:ach-b1}\\
&=e^{s\lambda_n} 
\exp\left(-sD_{1-s}(\rho_{AB}^{\otimes n}\| \rho_A^{\otimes n}\otimes\omega_{B^n}^n )\right)
\\
&=e^{-nR}.
\label{eq:ach-b3}
\end{align}
\end{subequations}
\eqref{eq:ach-b1} follows from \cite[Eq.~(2.2)]{hayashi_2016-1}. 
\eqref{eq:ach-b3} follows from~\eqref{eq:singly-lambda}.
Furthermore, 
\begin{subequations}\label{eq:singly-ach-one}
\begin{align}
\sup_{\tau_{B^n}\in \mathcal{S}_{\sym}(B^n)}
\Tr[\rho_A^{\otimes n}\otimes \tau_{B^n} (1-T^n_{A^nB^n})]
&\leq g_{n,d_B}\Tr[\rho_A^{\otimes n}\otimes \omega_{B^n}^n (1-T^n_{A^nB^n})]
\label{eq:omega-bound1}\\
&=g_{n,d_B}\Tr[\rho_A^{\otimes n}\otimes \omega_{B^n}^n
\{e^{-\lambda_n}\rho_{AB}^{\otimes n}>\rho_A^{\otimes n}\otimes \omega_{B^n}^n\}]\\
&\leq g_{n,d_B}\Tr[(e^{-\lambda_n}\rho_{AB}^{\otimes n})^{1-s} (\rho_A^{\otimes n}\otimes \omega_{B^n}^n)^{s}]
\label{eq:singly-alpha-lambda}\\
&=g_{n,d_B}e^{(s-1)\lambda_n} 
\exp\left(-sD_{1-s}(\rho_{AB}^{\otimes n} \| \rho_A^{\otimes n}\otimes \omega_{B^n}^n)\right)\\
&=g_{n,d_B}\exp\left(
-D_{1-s}(\rho_{AB}^{\otimes n}\| \rho_A^{\otimes n} \otimes \omega_{B^n}^n) +\frac{1-s}{s} nR\right).
\label{eq:alpha-qn}
\end{align}
\end{subequations}
\eqref{eq:omega-bound1} holds because $\tau_{B^n}\leq g_{n,d_B}\omega_{B^n}^n$ for all $\tau_{B^n}\in \mathcal{S}_{\sym}(B^n)$~\cite{christandl2009postselection}. 
\eqref{eq:singly-alpha-lambda} follows from~\cite[Eq.~(2.2)]{hayashi_2016-1} and
\eqref{eq:alpha-qn} follows from~\eqref{eq:singly-lambda}. 
We conclude that 
\begin{subequations}\label{eq:ach-a}
\begin{align}
\liminf_{n\rightarrow\infty} -\frac{1}{n}\log \hat{\beta}_{n}^{\sym}(e^{-nR})
&\geq \liminf_{n\rightarrow\infty} -\frac{1}{n}\log \sup_{\tau_{B^n}\in \mathcal{S}_{\sym}(B^n)}\Tr[\rho_A^{\otimes n}\otimes \tau_{B^n}(1-T^n_{A^nB^n})]
\label{eq:ach-a0}\\
&\geq \liminf_{n\rightarrow\infty} \frac{1}{n} D_{1-s}(\rho_{AB}^{\otimes n} \| \rho_A^{\otimes n}\otimes \omega_{B^n}^n) -\frac{1-s}{s} R
\label{eq:ach-a1}\\
&= I_{1-s}^{\uparrow\downarrow}(A:B)_\rho -\frac{1-s}{s} R
\label{eq:ach-a2}\\
&= \frac{1-s}{s}(L_s^{\uparrow\downarrow}(A:B)_\rho -R).
\label{eq:ach-a3}
\end{align}
\end{subequations}
\eqref{eq:ach-a0} follows from~\eqref{eq:singly-feasible}. 
\eqref{eq:ach-a1} follows from~\eqref{eq:singly-ach-one} and \cite[Proposition~1~(b)]{burri2025doubly}.
\eqref{eq:ach-a2} follows from \cite[Proposition~4]{burri2025doubly}. 
Since $s\in (0,1)$ can be chosen arbitrarily, the assertion in~\eqref{eq:singly-ach} follows from~\eqref{eq:ach-a}.
\end{proof}

\subsubsection{Proof of optimality}\label{sec:singly_converse}
Below, we prove that for any $R\in (0,\infty)$
\begin{align}
\limsup_{n\rightarrow\infty}-\frac{1}{n}\log \hat{\beta}_{n}^{\mathrm{iid}}(e^{-nR})
\leq \sup_{s\in (0,1)}\frac{1-s}{s}(L_s^{\uparrow\downarrow}(A:B)_\rho - R).
\end{align}

\begin{proof}
Let $\tau_B\in \mathcal{S}(B)$ be arbitrary but fixed. 
We have
\begin{align}
    &\limsup_{n\rightarrow\infty}-\frac{1}{n}\log \hat{\beta}_{n}^{\mathrm{iid}}(e^{-nR}) \\
    &= \limsup_{n\rightarrow\infty}-\frac{1}{n}\log \inf_{0\leq T^n_{A^nB^n} \leq 1} 
    \{\sup_{\tilde{\tau}_B\in \mathcal{S}(B)}\Tr[\rho_A^{\otimes n}\otimes \tilde{\tau}_B^{\otimes n} (1-T^n_{A^nB^n})]:
    \alpha_n(T_{A^nB^n}^n)\leq e^{-nR}\} 
    \label{eq:hoeffding0}\\
    &\leq \limsup_{n\rightarrow\infty}-\frac{1}{n}\log \inf_{0\leq T^n_{A^nB^n} \leq 1} 
    \{\Tr[\rho_A^{\otimes n}\otimes \tau_B^{\otimes n} (1-T^n_{A^nB^n})]:  \alpha_n(T_{A^nB^n}^n)\leq e^{-nR}\} 
    \label{eq:hoeffding1}\\
    &= \sup_{t\in(0,1)} \frac{-tR-\log \Tr[\rho_A^{1-t}\otimes\tau_B^{1-t}\rho_{AB}^{t}]}{1-t}
    \label{eq:hoeffding2}\\
    &=\sup_{s\in(0,1)} \frac{1-s}{s} (D_s(\rho_A \otimes \tau_B \| \rho_{AB}) - R) \, .
    \label{eq:hoeffding3}
\end{align}
\eqref{eq:hoeffding2} follows from the converse of the quantum Hoeffding Bound~\cite{Nagaoka2006}. 
\eqref{eq:hoeffding3} follows from the previous line by introducing $s\coloneqq 1-t$. 
Since $\tau_B\in \mathcal{S}(B)$ was arbitrary, we can take the infimum over all such states. Thus, 
\begin{align}
    \limsup_{n\rightarrow\infty}-\frac{1}{n}\log \hat{\beta}_{n}^{\mathrm{iid}}(e^{-nR})
    &\leq \inf_{\tau_B\in \mathcal{S}(B)} \sup_{s\in(0,1)} \frac{1-s}{s} (D_s(\rho_A \otimes \tau_B \| \rho_{AB}) - R)
    \\
    &\leq \inf_{\tau_B\in \mathcal{S}_{\ll \rho_B}(B)} \sup_{s\in(0,1)} \frac{1-s}{s} (D_s(\rho_A \otimes \tau_B \| \rho_{AB}) - R)
    \\
    &= \sup_{s\in(0,1)} \inf_{\tau_B\in \mathcal{S}_{\ll \rho_B}(B)}\frac{1-s}{s} (D_s(\rho_A \otimes \tau_B \| \rho_{AB}) - R)
    \label{eq:minimax2}\\
    &= \sup_{s\in(0,1)} \frac{1-s}{s} (L_s^{\uparrow\downarrow}(A:B)_\rho - R)
    \,.
    \label{eq:minimax3}
\end{align}
\eqref{eq:minimax2} follows from the minimax theorem in \cite[Proposition~21]{hayashi_2016-1}. The conditions for applying this minimax theorem are met because the function
\begin{equation}
    (0,1)\times \mathcal{S}_{\ll \rho_B}(B)\rightarrow\mathbb{R} ,
    (s,\tau_B)\mapsto
    (1-s)D_s(\rho_a \otimes \tau_B \| \rho_{AB}) = -\log\Tr[\rho_A^s \otimes \tau_B^s \rho_{AB}^{1-s}]
\end{equation}
is convex in $\tau_B$ (because $X \to X^s$ is operator concave for $s\in(0,1)$) and therefore also $\frac{1}{2}$-convexlike. 
Furthermore, the expression is concave in $s$ due to~\cite[Lemma~2.1]{Audenaert2012_quantum}. 
Therefore, the conditions for applying the minimax theorem are met. 
\end{proof}

\subsection{Proof of Theorem~\ref{thm:doubly}}\label{proof:doubly}

Based on the three possible choices for the alternative hypothesis, 
let us define the following functions of $\mu\in [0,\infty)$ 
for any $\rho_{AB}\in \mathcal{S}(AB),n\in \mathbb{N}_{>0}$.
\begin{align}
\hat{\beta}_{n}^{\mathrm{iid}}(\mu )
&\coloneqq \min_{\substack{T^n_{A^nB^n}\in \mathcal{L} (A^nB^n): \\ 0\leq T^n_{A^nB^n}\leq 1}}
\{\max_{\substack{\sigma_A\in \mathcal{S}(A),\\ \tau_{B} \in \mathcal{S}(B)}}
\Tr[\sigma_A^{\otimes n}\otimes \tau_{B}^{\otimes n} T^n_{A^nB^n}]:
\Tr[\rho_{AB}^{\otimes n}(1-T^n_{A^nB^n})]
\leq \mu\}
\label{eq:def-beta-doubly1}
\\
\hat{\beta}_{n}^{\sym}(\mu )
&\coloneqq \min_{\substack{T^n_{A^nB^n}\in \mathcal{L} (A^nB^n): \\ 0\leq T^n_{A^nB^n}\leq 1}}
\{\max_{\substack{\sigma_{A^n}\in \mathcal{S}_{\sym}(A^n),\\ \tau_{B^n} \in \mathcal{S}_{\sym}(B^n)}}
\Tr[\sigma_{A^n}\otimes \tau_{B^n} T^n_{A^nB^n}]:
\Tr[\rho_{AB}^{\otimes n}(1-T^n_{A^nB^n})]
\leq \mu\}
\label{eq:def-beta-doubly2}
\\
\hat{\beta}_{n}^{\mathrm{sym,A}}(\mu )
&\coloneqq \min_{\substack{T^n_{A^nB^n}\in \mathcal{L} (A^nB^n): \\ 0\leq T^n_{A^nB^n}\leq 1}}
\{\max_{\substack{\sigma_{A^n}\in \mathcal{S}_{\sym}(A^n),\\ \tau_{B^n} \in \mathcal{S}(B^n)}}
\Tr[\sigma_{A^n}\otimes \tau_{B^n} T^n_{A^nB^n}]:
\Tr[\rho_{AB}^{\otimes n}(1-T^n_{A^nB^n})]
\leq \mu\}
\label{eq:def-beta-doubly3}
\end{align} 

The following lemma describes their natural ordering. 
\begin{lemma}[(Minimum type-I errors)]\label{lem:order_doubly}
Let $\rho_{AB}\in \mathcal{S}(AB),n\in \mathbb{N}_{>0}$. Then, for all $\mu\in [0,\infty)$
\begin{align}
    0&\leq \hat{\beta}_{n}^{\mathrm{iid}} (\mu)
\leq \hat{\beta}_{n}^{\sym}(\mu)
= \hat{\beta}_{n}^{\mathrm{sym,A}}(\mu)
\leq \max(0,1-\mu) \leq 1\,.
\end{align}
\end{lemma}
We omit a proof of this lemma since it can be proven completely analogous to Lemma~\ref{lem:order} (see also~\cite{burri2025doubly}).

The proof of Theorem~\ref{thm:doubly} is divided into two parts: a proof of achievability and a proof of optimality. Their combination with Lemma~\ref{lem:order_doubly} implies the claim.

\subsubsection{Proof of achievability}\label{sec:doubly_achievability}
Below, we prove that for any $R\in (0,\infty)$
\begin{align}\label{eq:doubly-ach}
\liminf_{n\rightarrow\infty}-\frac{1}{n}\log \hat{\beta}_{n}(e^{-nR})
\geq \sup_{s\in (0,1)}\frac{1-s}{s}(L_s^{\downarrow\downarrow}(A:B)_\rho - R).
\end{align}
\begin{proof}
Let $R\in (0,\infty)$ and $s\in (0,1)$ be arbitrary but fixed. For all $n\in \mathbb{N}_{>0}$, we define
\begin{align}
\lambda_n\coloneqq -\frac{1}{s} 
\left(\ nR-sD_{1-s}(\rho_{AB}^{\otimes n}\| \omega_{A^n}^{n} \otimes \omega_{B^n}^n )\right)
\label{eq:doubly-lambda}
\end{align}
and the test
$T^n_{A^nB^n}\coloneqq \{\rho_{AB}^{\otimes n}\leq e^{\lambda_n}\omega_{A^n}^n\otimes \omega_{B^n}^n\}$.
For this test holds
\begin{subequations}\label{eq:doubly-feasible}
\begin{align}
\Tr[\rho_{AB}^{\otimes n} T^n_{A^nB^n}]
&= \Tr[\rho_{AB}^{\otimes n}\{\rho_{AB}^{\otimes n}\leq e^{\lambda_n}\omega_{A^n}^n\otimes\omega_{B^n}^n \} ] 
\label{eq:doubly-ach-b0}\\
&\leq \Tr[(\rho_{AB}^{\otimes n})^{1-s} (e^{\lambda_n} \omega_{A^n}^n\otimes\omega_{B^n}^n )^{s}]
\label{eq:doubly-ach-b1}\\
&=e^{s\lambda_n} 
\exp\left(-sD_{1-s}(\rho_{AB}^{\otimes n}\| \omega_{A^n}^n\otimes\omega_{B^n}^n )\right)
\\
&=e^{-nR}.
\label{eq:doubly-ach-b3}
\end{align}
\end{subequations}
\eqref{eq:doubly-ach-b1} follows from \cite[Eq.~(2.2)]{hayashi_2016-1}. 
\eqref{eq:doubly-ach-b3} follows from~\eqref{eq:doubly-lambda}.
Furthermore,
\begin{subequations}\label{eq:doubly-ach-one}
\begin{align}
\sup_{\substack{\sigma_{A^n}\in \mathcal{S}_{\sym}(A^n), \\ \tau_{B^n}\in \mathcal{S}_{\sym}(B^n)}}
\Tr[\sigma_{A^n}\otimes \tau_{B^n} (1-T^n_{A^nB^n})]
&\leq g_{n,d_A}g_{n,d_B}\Tr[\omega_{A^n}^n\otimes \omega_{B^n}^n (1-T^n_{A^nB^n})]
\label{eq:omega-bound2}\\
&=g_{n,d_A}g_{n,d_B}\Tr[\omega_{A^n}^n\otimes \omega_{B^n}^n
\{e^{-\lambda_n}\rho_{AB}^{\otimes n}>\omega_{A^n}^n\otimes \omega_{B^n}^n\}]\\
&\leq g_{n,d_A}g_{n,d_B}\Tr[(e^{-\lambda_n}\rho_{AB}^{\otimes n})^{1-s} (\omega_{A^n}^n\otimes \omega_{B^n}^n)^{s}]
\label{eq:doubly-alpha-lambda}\\
&=g_{n,d_A}g_{n,d_B}e^{(s-1)\lambda_n} 
\exp\left(-sD_{1-s}(\rho_{AB}^{\otimes n} \| \omega_{A^n}^n\otimes \omega_{B^n}^n)\right)\\
&=g_{n,d_A}g_{n,d_B}\exp\left(
-D_{1-s}(\rho_{AB}^{\otimes n}\| \omega_{A^n}^n \otimes \omega_{B^n}^n) +\frac{1-s}{s} nR\right).
\label{eq:doubly-alpha-qn}
\end{align}
\end{subequations}
\eqref{eq:omega-bound2} holds because $\sigma_{A^n}\leq g_{n,d_A}\omega_{A^n}^n$ for all $\sigma_{A^n}\in \mathcal{S}_{\sym}(A^n)$, and similarly, $\tau_{B^n}\leq g_{n,d_B}\omega_{B^n}^n$ for all $\tau_{B^n}\in \mathcal{S}_{\sym}(B^n)$~\cite{christandl2009postselection}. 
\eqref{eq:doubly-alpha-lambda} follows from~\cite[Eq.~(2.2)]{hayashi_2016-1}.
\eqref{eq:doubly-alpha-qn} follows from~\eqref{eq:doubly-lambda}. 
We conclude that 
\begin{subequations}\label{eq:doubly-ach-a}
\begin{align}
\liminf_{n\rightarrow\infty} -\frac{1}{n}\log \hat{\beta}_{n}(e^{-nR})
&\geq \liminf_{n\rightarrow\infty} -\frac{1}{n}\log 
\sup_{\substack{\sigma_{A^n}\in \mathcal{S}_{\sym}(A^n), \\ \tau_{B^n}\in \mathcal{S}_{\sym}(B^n)}}
\Tr[\sigma_{A^n}\otimes \tau_{B^n}(1-T^n_{A^nB^n})]
\label{eq:doubly-ach-a0}\\
&\geq \liminf_{n\rightarrow\infty} \frac{1}{n} D_{1-s}(\rho_{AB}^{\otimes n} \| \omega_{A^n}^n\otimes \omega_{B^n}^n) -\frac{1-s}{s} R
\label{eq:doubly-ach-a1}\\
&= I_{1-s}^{\downarrow\downarrow}(A:B)_\rho -\frac{1-s}{s} R
\label{eq:doubly-ach-a2}\\
&= \frac{1-s}{s}(L_s^{\downarrow\downarrow}(A:B)_\rho -R).
\label{eq:doubly-ach-a3}
\end{align}
\end{subequations}
\eqref{eq:doubly-ach-a0} follows from~\eqref{eq:doubly-feasible}. 
\eqref{eq:doubly-ach-a1} follows from~\eqref{eq:doubly-ach-one} and \cite[Proposition~1~(b)]{burri2025doubly}.
\eqref{eq:doubly-ach-a2} follows from \cite[Proposition~4]{burri2025doubly}. 
\eqref{eq:doubly-ach-a3} follows from Theorem~\ref{thm:properties_petz}. 
Since $s\in (0,1)$ can be chosen arbitrarily, the assertion in~\eqref{eq:doubly-ach} follows from~\eqref{eq:doubly-ach-a}.
\end{proof}

\subsubsection{Proof of optimality}
Below, we prove that for any $R\in (0,R^L_{1/2}) \cup (T(A:B)_\rho,\infty)$ 

\begin{align}\label{eq:doubly-optimality}
\limsup_{n\rightarrow\infty}-\frac{1}{n}\log \hat{\beta}_{n}(e^{-nR})
\leq \sup_{s\in (0,1)}\frac{1-s}{s}(L_s^{\downarrow\downarrow}(A:B)_\rho - R).
\end{align}

\begin{proof}
Let $\sigma_A\in \mathcal{S}(A),\tau_B\in \mathcal{S}(B)$ be arbitrary but fixed. 
We have
\begin{align}
    &\limsup_{n\rightarrow\infty}-\frac{1}{n}\log \hat{\beta}_{n}^{\mathrm{iid}}(e^{-nR}) \\
    &= \limsup_{n\rightarrow\infty}-\frac{1}{n}\log \inf_{0\leq T^n_{A^nB^n} \leq 1} 
    \{\sup_{\substack{\tilde{\sigma}_A\in \mathcal{S}(A),\\ \tilde{\tau}_B\in \mathcal{S}(B)}}\Tr[\tilde{\sigma}_A^{\otimes n}\otimes \tilde{\tau}_B^{\otimes n} (1-T^n_{A^nB^n})]:  
    \alpha_n(T_{A^nB^n}^n)\leq e^{-nR}\} 
    \\
    &\leq \limsup_{n\rightarrow\infty}-\frac{1}{n}\log \inf_{0\leq T^n_{A^nB^n} \leq 1} 
    \{\Tr[\sigma_A^{\otimes n}\otimes \tau_B^{\otimes n} (1-T^n_{A^nB^n})]:  \alpha_n(T_{A^nB^n}^n)\leq e^{-nR}\} 
    \\
    &= \sup_{t\in(0,1)} \frac{-tR-\log \Tr[\sigma_A^{1-t}\otimes\tau_B^{1-t}\rho_{AB}^{t}]}{1-t}
    \label{eq:hoeffding2_doubly}\\
    &=\sup_{s\in(0,1)} \frac{1-s}{s} (D_s(\sigma_A \otimes \tau_B \| \rho_{AB}) - R) \, .
    \label{eq:hoeffding3_doubly}
\end{align}
\eqref{eq:hoeffding2_doubly} follows from the converse of the quantum Hoeffding Bound~\cite{Nagaoka2006}. 
\eqref{eq:hoeffding3_doubly} follows from the previous line by introducing $s\coloneqq 1-t$.

Let $R\in (0,R^L_{1/2}) \cup (T(A:B)_\rho,\infty)$.

\textbf{Case 1}: $R\in (T(A:B)_\rho,\infty)$. 
We can then assume that $T(A:B)_\rho<\infty$ (otherwise the claim for case 1 is void). 
Let $(\sigma_A^\star,\tau_B^\star )\in \argmin_{(\sigma_A,\tau_B)\in \mathcal{S}(A)\times \mathcal{S}(B)} D(\sigma_A\otimes \tau_B\| \rho_{AB})$. 
Then
\begin{align}
    \sup_{s\in(0,1)} \frac{1-s}{s} (D_s(\sigma_A^\star \otimes \tau_B^\star  \| \rho_{AB}) - R) &\leq \sup_{s\in(0,1)} \frac{1-s}{s} (D(\sigma_A^\star \otimes \tau_B^\star  \| \rho_{AB}) - R)
    \label{eq:case1_1}
    \\
    &= \sup_{s\in(0,1)} \frac{1-s}{s} (T(A:B)_\rho - R)=0
    \label{eq:case1_2}\\
    &= \lim_{s\to1^-} \frac{1-s}{s}(L_s^{\downarrow \downarrow}(A:B)_\rho-R) 
    \label{eq:case1_3}\\
    &\leq \sup_{s\in(0,1)} \frac{1-s}{s} (L_s^{\downarrow \downarrow}(A:B)_\rho-R) .
\end{align}
\eqref{eq:case1_1} holds due to monotonicity in the R\'enyi order of the Petz divergence. 
\eqref{eq:case1_2} holds because $T(A:B)_\rho<R$. 
\eqref{eq:case1_3} holds due to the continuity in $s$ of $L_s^{\downarrow\downarrow}(A:B)_\rho$, and its non-negativity.

\textbf{Case 2}: $R\in(0,R^L_{1/2})$. 
We can then assume that $R^L_{1/2}>0$ (otherwise the claim for case 2 is void). 
Let us define the following functions of $s\in (0,1/2)$. 
\begin{align}
    \phi(s) 
    &\coloneqq (s-1) L_s^{\downarrow \downarrow}(A:B)_\rho \\
    g(s) 
    &\coloneqq \frac{1-s}{s}(L_s^{\downarrow\downarrow}  (A:B)_\rho -R) = \frac{1}{s}((s-1)R-\phi(s))\\
    \psi(s) 
    &\coloneqq s\phi'(s) - \phi(s)
\end{align}
Due to shown properties, $\phi$ is convex and continuously differentiable. 
So also $g$ is continuously differentiable. 
It's derivative is 
$g'(s)=\frac{1}{s^2}(R-s\phi'(s)+\phi(s)$). 
This motivates the definition of the function $\psi$. 
Convexity of $\phi$ implies that $\psi$ is monotonically increasing (see for instance \cite[Lemma~18]{burri2025doubly}). 
Therefore $g'(s)=\frac{1}{s^2}(R-\psi(s))$ is monotonically decreasing. 
This means that $g$ is concave and its extremal point is actually a maximum. 
This maximum is achieved in $(0,1/2)$ because 
$R < R^L_{1/2}= \lim_{s\to1/2^-}\psi(s)$. 

Let $\hat{s}\in(0,1/2)$ be a maximizer of $g$, or equivalently, $g'(\hat{s})=0$. 
Let
\begin{align}
    (\sigma^{(\hat{s})}_A,\tau^{(\hat{s})}_B)\in \argmin_{(\sigma_A,\tau_B)\in \mathcal{S}(A)\times \mathcal{S}(B)}D_{\hat{s}}(\sigma_A\otimes\tau_B\|\rho_{AB})
\end{align}
be the unique minimizer (see Theorem~\ref{thm:doubly}). 
Then define the following functions of $s \in(0,1)$.
\begin{align}
    \bar\phi(s) 
    &\coloneqq (s-1) D_s(\sigma^{(\hat{s})}_A \otimes\tau^{(\hat{s})}_B \|\rho_{AB})\\
    \bar g(s) 
    &\coloneqq\frac{1-s}{s}(D_s(\sigma^{(\hat{s})}_A \otimes\tau^{(\hat{s})}_B \|\rho_{AB}) -R) = \frac{1}{s}((s-1)R-\bar\phi(s))\\
    \bar\psi(s) 
    &\coloneqq s\bar\phi'(s) - \bar\phi(s)
\end{align}
As before, $\bar\phi$ is convex and continuously differentiable. This leads to $\bar \psi$ being monotonically increasing, which in turn implies that $\bar g'(s) = \frac{1}{s^2}(R-\bar\psi(r))$ is monotonically decreasing. This implies $\bar g$ is concave with maximum for some $s \in (0,1)$ when $\bar g'(s)=0$. 
By choice of $(\sigma^{(\hat{s})}_A,\tau^{(\hat{s})}_B)$, we have $$L_{\hat{s}}^{\downarrow \downarrow}(A:B)_\rho = D_{\hat{s}}(\sigma^{(\hat{s})}_A \otimes\tau^{(\hat{s})}_B \|\rho_{AB})\, .$$ Furthermore, due to Theorem \ref{thm:properties} (l), the first derivative also satisfies $$\frac{\partial}{\partial s}L_s^{\downarrow\downarrow}(A:B)_\rho|_{s=\hat{s}} 
    =\frac{\partial}{\partial s} D_s(\sigma^{(\hat{s})}_A \otimes\tau^{(\hat{s})}_B \|\rho_{AB})|_{s=\hat{s}} \, .$$
This implies $\psi(\hat{s}) = \bar \psi(\hat{s})$, which means that $g'$ and $\bar g'$ have the same null at $s=\hat{s}$. Therefore we find for the maximum
\begin{align}
    \sup_{s\in(0,1)} \frac{1-s}{s} D_s(\sigma^{(\hat{s})}_A \otimes\tau^{(\hat{s})}_B \|\rho_{AB}) - R) &=  \frac{1-\hat{s}}{\hat{s}} D_{\hat{s}}(\sigma^{(\hat{s})}_A \otimes\tau^{(\hat{s})}_B \|\rho_{AB}) - R) \\
    &= \frac{1-\hat{s}}{\hat{s}} (L_{\hat{s}}^{\downarrow\downarrow}(A:B)_\rho - R) \\
    &\leq \sup_{s\in(0,1)} \frac{1-s}{s} (L_s^{\downarrow\downarrow}(A:B)_\rho - R) \, .
\end{align}
The proof the follows from picking $(\sigma_A,\tau_B) =(\sigma^{(\hat{s})}_A ,\tau^{(\hat{s})}_B)$ in equation \eqref{eq:hoeffding3_doubly}.

\end{proof}

\subsection{Examples for Remark~\ref{rem:r12l}: $R_{1/2}^L$ vs. $R_{1/2}$}\label{app:r12}
\begin{prop}
    Let $\rho_{AB}\in \mathcal{S}(AB)$ and let 
    \begin{align}
        R_{1/2} &\coloneqq I_{1/2}^{\downarrow\downarrow}(A:B)_\rho-\frac{1}{4}\frac{\partial}{\partial s^+}I_s^{\downarrow\downarrow}(A:B)_\rho |_{s=1/2}, 
        \\
        R^L_{1/2} &\coloneqq L_{1/2}^{\downarrow\downarrow}(A:B)_\rho-\frac{1}{4}\frac{\partial}{\partial s^-}L_s^{\downarrow\downarrow}(A:B)_\rho |_{s=1/2} .
    \end{align}
    Let $p_{\max}$ be the largest eigenvalue of $\rho_A$ and let $m$ be the multiplicity of $p_{\max}$. 
    \begin{enumerate}[label=(\alph*)]
        \item If $\rho_{AB}$ is a pure state, then 
        \begin{align}
            R_{1/2}&= -\log \left(\frac{p_{\max}}{m}\right),\\
            R_{1/2}^L&=-\log \left(mp_{\max}\right).
        \end{align}
        \item If $\rho_{AB}$ is a copy-CC state (see~\cite{burri2025doubly}), then 
        \begin{align}
            R_{1/2}&= \log \left(m\right),\\
            R_{1/2}^L&=-\log \left(mp_{\max}\right).
        \end{align}
    \end{enumerate}
\end{prop}
\begin{proof}
    Let us define the following functions.
    \begin{align}
        R:(1/2,1)\rightarrow\mathbb{R},
        s&\mapsto I_s^{\downarrow\downarrow}(A:B)_\rho -s(1-s)\frac{\dd}{\dd s}I_s^{\downarrow\downarrow}(A:B)_\rho,
        \\
        R^L:(0,1/2)\rightarrow\mathbb{R},
        s&\mapsto L_s^{\downarrow\downarrow}(A:B)_\rho -s(1-s)\frac{\dd}{\dd s}L_s^{\downarrow\downarrow}(A:B)_\rho.
    \end{align}
    In the following, we will use that for all $\alpha\in (0,1)\cup (1,\infty)$
    \begin{align}
        \frac{\dd}{\dd \alpha}H_\alpha (A)_\rho =-\frac{1}{(1-\alpha)^2} D(\sigma_A^{(\alpha)}\| \rho_A) 
        \qquad \text{where}\qquad
        \sigma_A^{(\alpha)}\coloneqq \rho_A^\alpha / \Tr[\rho_A^\alpha]. 
    \end{align}
    \textbf{Proof of (a).} 
    We have for all $s\in (1/2,1)$
    \begin{align}
        R(s)&=2H_{\frac{1}{2s-1}}(A)_\rho +\frac{4s(1-s)}{(2s-1)^2}\frac{\dd}{\dd \alpha}H_\alpha (A)_\rho \big|_{\alpha = \frac{1}{2s-1}}
        \\
        &=2H_{\frac{1}{2s-1}}(A)_\rho +\frac{s}{s-1}D(\sigma_A^{(\alpha)}\| \rho_A)\big|_{\alpha = \frac{1}{2s-1}}.
    \end{align}
    Thus,
    \begin{align}
        R_{1/2}&=\lim_{s\rightarrow 1/2^+}R(s)
        =2H_\infty (A)_\rho -\lim_{\alpha\rightarrow\infty}D(\sigma_A^{(\alpha)}\| \rho_A)
        =-2\log (p_{\max}) +\log (mp_{\max})
        =-\log (p_{\max} /m).
    \end{align}
    We have for all $s\in (0,1/2)$
    \begin{align}
        R^L(s)&=s^2\frac{\dd}{\dd t}I_t^{\downarrow\downarrow}(A:B)_\rho|_{t=1-s}
        =s^2\frac{\dd}{\dd t}2H_{\frac{1}{2t-1}}(A)_\rho |_{t=1-s}
        \\
        &=\frac{-4s^2}{(2t-1)^2}\frac{\dd}{\dd \alpha}H_\alpha (A)_\rho\big|_{\alpha=\frac{1}{2t-1},t=1-s}
        =D(\sigma_A^{(\alpha)}\| \rho_A)\big|_{\alpha=\frac{1}{2t-1},t=1-s}
    \end{align}
    Thus,
    \begin{align}
        R_{1/2}^L &=\lim_{s\rightarrow 1/2^-}R^L(s)
        =\lim_{\alpha\rightarrow\infty}D(\sigma_A^{(\alpha)}\| \rho_A)
        =-\log(mp_{\max}) \, 
    \end{align}
\textbf{Proof of (b).} 
    We have for all $s\in (1/2,1)$
    \begin{align}
        R(s)&=H_{\frac{s}{2s-1}}(A)_\rho 
        +\frac{s(1-s)}{(2s-1)^2}\frac{\dd}{\dd \alpha}H_\alpha (A)_\rho \big|_{\alpha = \frac{s}{2s-1}}
        \\
        &=H_{\frac{s}{2s-1}}(A)_\rho -\frac{s}{1-s}D(\sigma_A^{(\alpha)}\| \rho_A).
    \end{align}
    Thus,
    \begin{align}
        R_{1/2}&=\lim_{s\rightarrow 1/2^+}R(s)
        =H_\infty (A)_\rho -\lim_{\alpha\rightarrow\infty}D(\sigma_A^{(\alpha)}\| \rho_A)
        =-\log (p_{\max}) +\log (mp_{\max})
        =\log (m).
    \end{align}
    We have for all $s\in (0,1/2)$
    \begin{align}
        R^L(s)&=s^2\frac{\dd}{\dd t}I_t^{\downarrow\downarrow}(A:B)_\rho|_{t=1-s}
        =s^2\frac{\dd}{\dd t}H_{\frac{t}{2t-1}}(A)_\rho \big|_{t=1-s}
        \\
        &=\frac{-s^2}{(2t-1)^2}\frac{\dd}{\dd \alpha}H_\alpha (A)_\rho\big|_{\alpha=\frac{t}{2t-1},t=1-s}
        =D(\sigma_A^{(\alpha)}\| \rho_A)\, ,
    \end{align}
    thus,
    \begin{align}
        R_{1/2}^L &=\lim_{s\rightarrow 1/2^-}R^L(s)
        =\lim_{\alpha\rightarrow\infty}D(\sigma_A^{(\alpha)}\| \rho_A)
        =-\log(mp_{\max})  .
    \end{align}
\end{proof}

\subsection{Proof of Corollary~\ref{cor:sanov}}\label{proof:sanov}

\begin{proof} 
Let $\varepsilon\in (0,1)$. 
The claim follows from the following chain of inequalities.
\begin{align}
    T(A:B)_\rho &=\inf_{\substack{\sigma_A\in \mathcal{S}(A),\\ \tau_B\in \mathcal{S}(B)}}D(\sigma_A\otimes \tau_B\| \rho_{AB})
    \label{eq:proof_sanov0}\\
    &=\mathrm{Sanov}_\varepsilon((\rho_{AB}^{\otimes n})_{n\in \mathbb{N}>0}\| (\{\sigma_{A}^{\otimes n}\otimes \tau_{B}^{\otimes n}\}_{\sigma_{A}\in \mathcal{S}(A),\tau_{B}\in \mathcal{S}(B)})_{n\in \mathbb{N}_{>0}})
    \label{eq:proof_sanov1}\\
    &\geq \mathrm{Sanov}_\varepsilon((\rho_{AB}^{\otimes n})_{n\in \mathbb{N}>0}\| (\{\sigma_{A^n}\otimes \tau_{B^n}\}_{\sigma_{A^n}\in \mathcal{S}_{\sym}(A^n),\tau_{B^n}\in \mathcal{S}_{\sym}(B^n)})_{n\in \mathbb{N}_{>0}})
    \label{eq:proof_sanov2}\\
    &\geq T(A:B)_\rho
    \label{eq:proof_sanov21}
\end{align}
\eqref{eq:proof_sanov2} follows from Sanov's theorem for composite iid-hypothesis testing~\cite{berta_composite,Bjelakovic2005,Mosonyi_2015,Noetzel_2014}. 
\eqref{eq:proof_sanov2} holds because $\sigma_{A}^{\otimes n}\in \mathcal{S}_{\sym}(A^n), \tau_{B}^{\otimes n}\in \mathcal{S}_{\sym}(B^n)$ for all $n\in \mathbb{N}_{>0}$. 
It remains to prove~\eqref{eq:proof_sanov21}. 
Accordingly, we consider 
$H_1^n\coloneqq \{\sigma_{A^n}\otimes \tau_{B^n}\}_{\sigma_{A^n}\in \mathcal{S}_{\sym}(A^n),\tau_{B^n}\in \mathcal{S}_{\sym}(B^n)}$ in the remainder of this proof. 

\textbf{Case 1}: $\rho_{AB}\neq \rho_A\otimes\rho_B$. 
Let $R\in (0,T(A:B)_\rho)$. 
Let $s$ be the corresponding optimizer of the right-hand side of~\eqref{eq:doubly-ach}. 
Consider the tests $T_{A^nB^n}^n$ defined as in the proof of~\eqref{eq:doubly-ach} for this $s$. Then,

\begin{align}
    &\mathrm{Sanov}_\varepsilon((\rho_{AB}^{\otimes n})_{n\in \mathbb{N}>0}\| (\{\sigma_{A^n}\otimes \tau_{B^n}\}_{\sigma_{A^n}\in \mathcal{S}_{\sym}(A^n),\tau_{B^n}\in \mathcal{S}_{\sym}(B^n)})_{n\in \mathbb{N}_{>0}})
    \label{eq:proof_sanov3}
    \\
    &\geq \liminf_{n\rightarrow\infty}-\frac{1}{n}\log \inf_{\substack{\tilde{T}_{A^nB^n}^n\in\mathcal{L}(A^nB^n):\\ 0\leq \tilde{T}_{A^nB^n}^n\leq \id}} 
    \{\Tr[\rho_{AB}^{\otimes n}\tilde{T}_{A^nB^n}^n] :
    \sup_{\substack{\sigma_{A^n}\in \mathcal{S}_{\sym}(A^n),\\ \tau_{B^n}\in \mathcal{S}_{\sym}(B^n)}}
    \Tr[\sigma_{A^n}\otimes \tau_{B^n}(1-\tilde{T}_{A^nB^n}^n)]\leq \varepsilon\}
    \label{eq:proof_sanov4}
    \\
    &\geq \liminf_{n\rightarrow\infty} -\frac{1}{n}\log \Tr[\rho_{AB}^{\otimes n}T_{A^nB^n}^n]
    \label{eq:proof_sanov5}\\
    &\geq \liminf_{n\rightarrow\infty} -\frac{1}{n}\log e^{-nR}
    =R
    \label{eq:proof_sanov6}
\end{align}
\eqref{eq:proof_sanov5} holds because the test $T_{A^nB^n}^n$ is in the feasible set, as the proof of~\eqref{eq:doubly-ach} implies that 
$\Tr[\sigma_{A^n}\otimes \tau_{B^n}(1-\tilde{T}_{A^nB^n}^n)]\rightarrow 0$ as $n\rightarrow\infty$, because the right-hand side of~\eqref{eq:doubly-ach} is strictly positive as $\rho_{AB}\neq \rho_A\otimes \rho_B$. 
\eqref{eq:proof_sanov6} follows from~\eqref{eq:doubly-feasible}.
Since $R\in (0,T(A:B)_\rho)$ was arbitrary, we can conclude that~\eqref{eq:proof_sanov21} holds.

\textbf{Case 2}: $\rho_{AB}= \rho_A\otimes\rho_B$. 
Then, $T(A:B)_\rho=0$. 
As the Sanov exponent is non-negative for any $\varepsilon\in (0,1)$, the inequality in~\eqref{eq:proof_sanov21} is trivially true.
\end{proof}

\subsection{Proof of Corollary~\ref{cor:lowrate}}\label{proof:lowrate}
\begin{proof}[Proof of (a)]
The first row of Table~\ref{tab:direct_prmi} implies that 
\begin{align}
\lim_{R\to 0^+}\lim_{n\rightarrow\infty}-\frac{1}{n}\log \hat{\alpha}_{n}(e^{-nR})
&= \sup_{R>0}\sup_{s\in (0,1)}\frac{1-s}{s}(I_s^{\uparrow\uparrow}(A:B)_\rho - R)\\
&= \sup_{s\in (0,1)}\sup_{R>0}\frac{1-s}{s}(I_s^{\uparrow\uparrow}(A:B)_\rho - R)\\
&\eqt{\eqref{eq:prli_prmi0}}\sup_{s\in (0,1)} L_{1-s}^{\uparrow\uparrow}(A:B)_\rho 
=\lim_{s\to 0} L_{1-s}^{\uparrow\uparrow}(A:B)_\rho 
=L(A:B)_\rho.
\end{align}
\end{proof}
\begin{proof}[Proof of (b), (c)]
These assertions follow analogously by using the second and third row of Table~\ref{tab:direct_prmi} instead of the first row, and by using~\eqref{eq:prli_prmi1} and \eqref{eq:prli_prmi2} instead of~\eqref{eq:prli_prmi0}, respectively.
\end{proof}
\begin{proof}[Proof of (d)]
If $H_1^n$ is defined as in (a), then the first row of Table~\ref{tab:direct_prli} implies that 
\bb 
\lim_{R\to 0^+}\lim_{n\rightarrow\infty}-\frac{1}{n}\log \hat{\beta}_{n}(e^{-nR})
&= \sup_{R>0}\sup_{s\in (0,1)}\frac{1-s}{s}(L_s^{\uparrow\uparrow}(A:B)_\rho - R)\\
&= \sup_{s\in (0,1)}\sup_{R>0}\frac{1-s}{s}(L_s^{\uparrow\uparrow}(A:B)_\rho - R)\\
&\eqt{\eqref{eq:prli_prmi0}} \sup_{s\in (0,1)}I_{1-s}^{\uparrow\uparrow}(A:B)_\rho 
=\lim_{s\to 0} I_{1-s}^{\uparrow\uparrow}(A:B)_\rho 
=I(A:B)_\rho.
\ee
If $H_1^n$ is defined as in (b), then the second row of Table~\ref{tab:direct_prli} implies the claim analogously by using~\eqref{eq:prli_prmi1} instead of~\eqref{eq:prli_prmi0}.

If $H_1^n$ is defined as in (c), then the third row of Table~\ref{tab:direct_prli} implies the claim analogously by using~\eqref{eq:prli_prmi2} instead of~\eqref{eq:prli_prmi0}.
\end{proof}

\section{Proofs for channel tumula information}

\subsection{Proof of Proposition~\ref{prop:superadditivity}}\label{proof:superadditivity}
\begin{proof}
    \begin{align}
        T(\pazocal{N}_1\otimes \pazocal{N}_2) &= \sup_{\Psi_{A_1'A_2'A_1A_2}} T(A'_1,A'_2:B_1,B_2)_{(\mathrm{Id}_{12}\otimes\pazocal{N}_1\otimes\pazocal{N}_2)(\Psi_{12})}\\
        &\geq \sup_{\Psi_{A_1'A_1}\otimes\Psi_{A_2'A_2}} T(A'_1,A'_2:B_1,B_2)_{(\mathrm{Id}_{1}\otimes\pazocal{N}_1)(\Psi_1)\otimes (\mathrm{Id}_{2}\otimes\pazocal{N}_2)(\Psi_2)}\\
        &=\sup_{\Psi_{A_1'A_1}} T(A'_1:B_1)_{(\mathrm{Id}_{1}\otimes\pazocal{N}_1)(\Psi_1)}+\sup_{\Psi_{A_2'A_2}} T(A'_2:B_2)_{(\mathrm{Id}_{2}\otimes\pazocal{N}_2)(\Psi_2)}
        \label{eq:proof_superadditive}\\
        &=T(\pazocal{N}_1)+T(\pazocal{N}_2)
    \end{align}
    In~\eqref{eq:proof_superadditive} we have leveraged the additivity of the tumula information for product states.
\end{proof}

\subsection{Proof of Theorem \ref{thm:CQ}}\label{proof:CQ}

\begin{proof}
    For a CQ-channel $\pazocal{N}$, we can write its tumula information as
    \begin{equation}
        T(\pazocal{N}) = \sup_{P_X} \min_{\sigma_A,\tau_B} D(\sigma_A\otimes\tau_B\|\sum_{x\in \mathcal{X}} P_X(x)\ketbra{x}\otimes \rho_x)
    \end{equation}
    Choosing the partial minimizer for $\tau_B$ (Theorem \ref{thm:properties})
    
    $$ \tau_B^\star  = \frac{e^{\Tr_A[\sigma_A \log \sum_{x\in \mathcal{X}} P_X(x)\ketbra{x} \otimes \rho_x]}}{\Tr[e^{\Tr_A[\sigma_A \log \sum_{x\in \mathcal{X}} P_X(x)\ketbra{x} \otimes \rho_x]}]}$$
    
    Introducing $Q_X(x) = \bra{x}\sigma_A\ket{x}$, this expression can be written as 
    
    $$
    \tau_B^\star  = \frac{e^{\sum_{x\in \mathcal{X}}Q_X(x)\log P_X(x)+\sum_{x\in \mathcal{X}} Q_X(x) \log(\rho_x)}}{\Tr[e^{\sum_{x\in \mathcal{X}}Q_X(x)\log P_X(x)+\sum_{x\in \mathcal{X}} Q_X(x) \log(\rho_x)}]} = \frac{e^{\sum_{x\in \mathcal{X}} Q_X(x) \log(\rho_x)}}{\Tr[e^{\sum_{x\in \mathcal{X}} Q_X(x) \log(\rho_x)}]} = \frac{e^L}{Z(Q_X)}
    $$
    where we introduced 
    $L\coloneqq\sum_{x\in \mathcal{X}}Q_X(x)\log\rho_x$ and 
    $Z(Q_X) \coloneqq \Tr[e^{\sum_{x\in \mathcal{X}} Q_X(x) \log(\rho_x)}]$.
    Using this expression leads to
    \begin{align}
        & D\left(\sigma_A\otimes\tau_B^\star \|\sum_{x\in \mathcal{X}} P_X(x)\ketbra{x} \otimes \rho_x\right) \\
        &= -H(B)_\sigma-H(B)_{\tau^\star} -\Tr[\sigma_A\otimes\tau^\star_B \sum_{x\in \mathcal{X}} P_X(x)\ketbra{x} \otimes \log\rho_x] \\
        &= -H(B)_\sigma-H(B)_{\tau^\star} -\sum_{x\in \mathcal{X}} Q_X(x) (\log P_X(x) + \Tr[\tau^\star_B\log\rho_x]) \\
        &= -H(B)_\sigma +\Tr[\tau_B^\star L] - \log Z(Q_X) - \sum_{x\in \mathcal{X}} Q_X(x) \log P_X(x) - \Tr[\tau_B^\star L] \\
        &= -H(B)_\sigma - \log Z(Q_X) - \sum_{x\in \mathcal{X}} Q_X(x) \log P_X(x)
    \end{align}
    Using the Gibbs expression for $\sigma_A^\star$ (Theorem \ref{thm:properties}) 
    \begin{equation}
        \sigma^*_A \sim e^{\Tr_B[\tau_B^\star  \sum_{x\in \mathcal{X}}P_X(x)\ketbra{x} \otimes \rho_x]} 
    \end{equation}
    we find that $\sigma^*$ is diagonal in the basis $(\ket{x})_{x\in \mathcal{X}}$ and we can write it as
    $$
    \sigma_A = \sum_{x\in \mathcal{X}} Q_X(x) \ketbra{x}_A \, .
    $$
    Altogether, this leads to
    \begin{equation}
        T(\pazocal{N}) = \sup_{P_X} \min_{Q_X} D(Q_X\|P_X) - \log Z(Q_X).
    \end{equation}
\end{proof}

\subsection{Proof of Proposition \ref{prop:TU_CQ}}\label{proof:TU_CQ}

\begin{proof}
    Define $L\coloneqq\sum_{x\in \mathcal{X}} Q_X(x)\log\rho_x$ and 
    $Z(Q_X) \coloneqq \Tr[e^{\sum_{x\in \mathcal{X}} Q_X(x) \log(\rho_x)}]$, then we can write
    \begin{align}
        \log Z(Q_X) = \max_\sigma \{ \Tr[\sigma L] + H(B)_\sigma \} \, .
    \end{align}
    To see this, consider $\tau = \frac{e^{L}}{Z(Q_X)}$. For any state $\sigma$ we have
    \begin{align}
        0 \leq D(\sigma\|\tau) &= \Tr[\sigma (\log \sigma -\log\tau)] \\
        &= -H(B)_\sigma - \Tr(\sigma L) + \log Z
    \end{align}
    where equality is achieved if $\sigma = \tau.$
    We then have
    \begin{equation}
        - \log Z(Q_X) = \min_\sigma \{ -\Tr[\sigma L] - H(B)_\sigma \}
    \end{equation} 
    Which leads to
    \begin{align}
        \min_{Q_X} \{ D(Q_X\|P_X) -\log Z(Q_X) \} 
        &= \min_{Q_X} \min_\sigma \left( D(Q_X\|P_X) - \Tr[\sigma \sum_{x\in \mathcal{X}} Q_X(x) \log \rho_x] - H(B)_\sigma \right)\\
        &= \min_\sigma \left( \min_{Q_X} \left( D(Q_X\|P_X) - \Tr[\sigma \sum_{x\in \mathcal{X}} Q_X(x) \log \rho_x] \right)- H(B)_\sigma \right) \, 
    \end{align}
    where we could change the order of the minimization since we have compactness and continuity. Now, define $s(x) = \Tr[\sigma \log\rho_x]$. Then for fixed $\sigma$, we have
    \begin{align}
        \min_{Q_X} \left\{ D(Q_X\|P_X) - \sum_{x\in \mathcal{X}} Q_X(x)s(x)\right\} &= \min_{Q_X}  \sum_{x\in \mathcal{X}} Q_X(x) \log\left(\frac{Q_X(x)}{P_X(x)e^{s(x)}}\right) \\
        &= \min_{Q_X} D(Q_X\|P'_X) - \log \sum_{x\in \mathcal{X}} P_X(x) e^{s(x)} \\
        &= -\log \sum_{x\in \mathcal{X}} P_X(x) e^{\Tr[\sigma \log\rho_x]}
    \end{align}
    where we introduced $P_X'(x) = \frac{P_X(x)e^{s(x)}}{\sum_{y\in \mathcal{X}} P_X(y)e^{s(y)}}$.
    Putting these steps together, we find
    \begin{align}
        T(\pazocal{N}) &= \sup_{P_X}\min_{Q_X} \{ D(Q_X\|P_X) - \log Z(Q_X) \} 
        \\ &= \sup_{P_X} \min_\sigma \left\{ -\log \sum_{x\in \mathcal{X}}P_X(x)e^{\Tr[\sigma \log \rho_x]}-H(B)_\sigma \right\} \\
        &= \sup_{P_X} \min_\sigma \left\{ -\log \sum_{x\in \mathcal{X}}P_X(x)e^{-D(\sigma\|\rho_x)-H(B)_\sigma} - H(B)_\sigma  \right\} \\
        &= \sup_{P_X} \min_\sigma -\log \sum_{x\in \mathcal{X}} P_X(x) e^{-D(\sigma\|\rho_x)}
    \end{align}
    If we compare this with Umlaut information of a CQ channel, we have
    \begin{align}
        U(\pazocal{N}) &= \sup_{P_X} -\log Z(P_X) \\
        &= \sup_{P_X} \min_\sigma \left\{ -\sum_{x\in \mathcal{X}} P_X(x)\Tr[\sigma \log \rho_x] - H(B)_\sigma \right\}\\
        &=  \sup_{P_X} \min_\sigma  \sum_{x\in \mathcal{X}} P_X(x) D(\sigma\|\rho_x) 
    \end{align}
    Both expressions are connected by Jensen's inequality, which is in general strict.
\end{proof}

\subsection{Proof of Proposition~\ref{prop:superadditivity_classical}}\label{proof:superadditivity_classical}
\begin{proof}
    \begin{align}
        T(\pazocal{W}_1\times \pazocal{W}_2) 
        &= \sup_{P_{X_1X_2}} T(X_1X_2:Y_1Y_2)_{(\pazocal{W}_1)_{Y_1|X_1} (\pazocal{W}_2)_{Y_2|X_2} P_{X_1X_2}}\\
        &\geq \sup_{P_{X_1}P_{X_2}} T(X_1X_2:Y_1Y_2)_{(\pazocal{W}_1)_{Y_1|X_1} P_{X_1}(\pazocal{W}_2)_{Y_2|X_2} P_{X_2}}\\
        &= \sup_{P_{X_1}} T(X_1:Y_1)_{(\pazocal{W}_1)_{Y_1|X_1} P_{X_1}}
        +\sup_{P_{X_2}} T(X_2:Y_2)_{(\pazocal{W}_2)_{Y_2|X_2} P_{X_2}}
        \label{eq:proof_superadditivity_classical}\\
        &=T(\pazocal{W}_1)+T(\pazocal{W}_2)
    \end{align}
    In~\eqref{eq:proof_superadditivity_classical} we have leveraged the additivity of the tumula information for classical product states, see Theorem~\ref{thm:properties_petz}~(n) and Theorem~\ref{thm:properties}~(b).
\end{proof}

\subsection{Proof of Theorem \ref{thm:tu_id_channel}}\label{proof:id}
\begin{proof}
    By Proposition~\ref{prop:TU_CQ}, we can write 
    \begin{align}
        T^\infty(I) &= \lim_{n \to \infty} \frac{1}{n} \sup_{P_{X^n}} \min_{Q_{X^n}} \left( - \log \sum_{x^n} P_{X^n}(x^n) e^{-D(Q_{X^n}\|I_{x^n})}\right) \, ,
    \end{align}
    where $I_{x^n}$ denotes the deterministic conditional distribution with $\delta_{x^n,y^n}$. To avoid the expression from diverging to infinity, $Q_{X^n}$ has to be deterministically supported on exactly one sequence. If $Q_{X^n}$ had support on two or more sequences, each divergence would be infinite, leading to a diverging term. Since $Q_{X^n}$ is deterministic for a specific sequence, the expression simplifies to
    \begin{align}
        T^\infty(I) &= \lim_{n \to \infty} \frac{1}{n} \sup_{P_{X^n}} \min_{x^n} \left( - \log P_{X^n}(x^n) \right) 
    \end{align}
    Since for every distribution $P_{X^n}$, there exists a sequence $x^n$ with $P_{X^n}(x^n)\geq \frac{1}{|\mathcal{X}|^n}$, the optimization is solved by
    \begin{align}
        T^\infty(I) &= \log|\mathcal{X}| \,.
    \end{align}
    The proof for $T(I)$ is completely analogous.
\end{proof}

\bibliography{biblio}

\end{document}

%% file: macros.tex
\usepackage{newpxtext,newpxmath}

\let\coloneqq\relax

\usepackage{amsthm}
\usepackage{amssymb}
\usepackage{amsmath}
\usepackage{bbold}
\usepackage{bbm}
\usepackage[pdftex, backref=page]{hyperref}
\usepackage{braket}
\usepackage{dsfont}
\usepackage{mathdots}
\usepackage{mathtools}
\usepackage{enumerate}
\usepackage[shortlabels]{enumitem}
\usepackage{csquotes}
\usepackage{stmaryrd}
\usepackage[cal=boondox]{mathalfa}
\usepackage{graphicx}
\usepackage{stackengine}
\usepackage{tensor}
\usepackage[dvipsnames]{xcolor}

\usepackage{tikz}
\usepackage{pgfplots}
\usetikzlibrary{shapes.geometric, shapes.misc, positioning, arrows, arrows.meta, decorations.pathreplacing, decorations.pathmorphing, patterns, angles, quotes, calc}
\usepackage{booktabs}
\usepackage{xfrac}
\usepackage{siunitx}
\usepackage{centernot}
\usepackage{comment}

\newtheorem{thm}{Theorem}
\newtheorem*{thm*}{Theorem}
\newtheorem{prop}[thm]{Proposition}
\newtheorem*{prop*}{Proposition}
\newtheorem{lemma}[thm]{Lemma}
\newtheorem*{lemma*}{Lemma}
\newtheorem{cor}[thm]{Corollary}
\newtheorem*{cor*}{Corollary}

\newtheorem*{cj*}{Conjecture}
\newtheorem{Def}[thm]{Definition}
\newtheorem*{Def*}{Definition}

\newtheorem*{question*}{Question}

\newtheorem*{problem*}{Problem}

\makeatletter
\def\thmhead@plain#1#2#3{%
  \thmname{#1}\thmnumber{\@ifnotempty{#1}{ }\@upn{#2}}%
  \thmnote{ {\the\thm@notefont#3}}}
\let\thmhead\thmhead@plain
\makeatother

\theoremstyle{definition}
\newtheorem{rem}[thm]{Remark}

\hypersetup{
    bookmarksnumbered=true, 
    breaklinks=true,
    unicode=false, 
    pdfstartview={FitH}, 
    pdfnewwindow=true, 
    colorlinks=true, 
    allcolors=MidnightBlue!70!black!70!TealBlue
}

\newcommand{\bb}{\begin{equation}\begin{aligned}\hspace{0pt}}
\newcommand{\bbb}{\begin{equation*}\begin{aligned}}
\newcommand{\ee}{\end{aligned}\end{equation}}
\newcommand{\eee}{\end{aligned}\end{equation*}}
\newcommand*{\coloneqq}{\mathrel{\vcenter{\baselineskip0.5ex \lineskiplimit0pt \hbox{\scriptsize.}\hbox{\scriptsize.}}} =}

\newcommand{\eqt}[1]{\stackrel{\mathclap{\scriptsize \mbox{#1}}}{=}}
\newcommand{\leqt}[1]{\stackrel{\mathclap{\scriptsize \mbox{#1}}}{\leq}}

\newcommand{\ketbra}[1]{\ket{#1}\!\!\bra{#1}}

\newcommand{\ketbraa}[2]{\ket{#1}\!\!\bra{#2}}

\renewcommand{\epsilon}{\varepsilon}

\newcommand{\dd}{\mathrm{d}}
\newcommand{\id}{\mathds{1}}

\newcommand{\cptp}{\mathrm{CPTP}}

\DeclareMathOperator*{\argmin}{\arg\min}
\DeclareMathOperator{\Tr}{Tr}
\DeclareMathAlphabet{\pazocal}{OMS}{zplm}{m}{n}

\DeclareMathOperator{\sym}{sym}

\stackMath

\stackMath

\makeatletter
\newcommand*\rel@kern[1]{\kern#1\dimexpr\macc@kerna}
\newcommand*\widebar[1]{%
  \begingroup
  \def\mathaccent##1##2{%
    \rel@kern{0.8}%
    \overline{\rel@kern{-0.8}\macc@nucleus\rel@kern{0.2}}%
    \rel@kern{-0.2}%
  }%
  \macc@depth\@ne
  \let\math@bgroup\@empty \let\math@egroup\macc@set@skewchar
  \mathsurround\z@ \frozen@everymath{\mathgroup\macc@group\relax}%
  \macc@set@skewchar\relax
  \let\mathaccentV\macc@nested@a
  \macc@nested@a\relax111{#1}%
  \endgroup
}

\counterwithin*{equation}{part}
\counterwithin*{thm}{part}
\counterwithin*{figure}{part}

\tikzset{meter/.append style={draw, inner sep=10, rectangle, font=\vphantom{A}, minimum width=30, line width=.8, path picture={\draw[black] ([shift={(.1,.3)}]path picture bounding box.south west) to[bend left=50] ([shift={(-.1,.3)}]path picture bounding box.south east);\draw[black,-latex] ([shift={(0,.1)}]path picture bounding box.south) -- ([shift={(.3,-.1)}]path picture bounding box.north);}}}
\tikzset{roundnode/.append style={circle, draw=black, fill=gray!20, thick, minimum size=10mm}}
\tikzset{squarenode/.style={rectangle, draw=black, fill=none, thick, minimum size=10mm}}

\definecolor{Blues5seq1}{RGB}{239,243,255}
\definecolor{Blues5seq2}{RGB}{189,215,231}
\definecolor{Blues5seq3}{RGB}{107,174,214}
\definecolor{Blues5seq4}{RGB}{49,130,189}
\definecolor{Blues5seq5}{RGB}{8,81,156}

\definecolor{Greens5seq1}{RGB}{237,248,233}
\definecolor{Greens5seq2}{RGB}{186,228,179}
\definecolor{Greens5seq3}{RGB}{116,196,118}
\definecolor{Greens5seq4}{RGB}{49,163,84}
\definecolor{Greens5seq5}{RGB}{0,109,44}

\definecolor{Reds5seq1}{RGB}{254,229,217}
\definecolor{Reds5seq2}{RGB}{252,174,145}
\definecolor{Reds5seq3}{RGB}{251,106,74}
\definecolor{Reds5seq4}{RGB}{222,45,38}
\definecolor{Reds5seq5}{RGB}{165,15,21}

\allowdisplaybreaks

\usepackage[most,breakable]{tcolorbox}
\renewenvironment{boxed}[1]%
	{\expandafter\ifstrequal\expandafter{#1}{orange}{\begin{tcolorbox}[colback=red!15,colframe=orange!15,breakable,enhanced]}{\begin{tcolorbox}[colback=Blues5seq1,colframe=Blues5seq5,breakable,enhanced]}}%
	{\end{tcolorbox}}